\documentclass[sigconf,nonacm]{acmart}

\setcopyright{none}
\usepackage{booktabs}
\usepackage{multirow}
\usepackage{graphicx}
\usepackage{subcaption}
\usepackage{url}
\usepackage{hyperref}
\usepackage{xcolor}
\usepackage{float}
\usepackage{tikz}
\usepackage{amsthm}
\newtheorem{remark}{Remark}
\usetikzlibrary{automata,arrows,positioning,shapes,arrows.meta,shapes.arrows}
\usetikzlibrary{calc}
\usepackage{acro}
\usepackage{stmaryrd}
\acsetup{
  first-style = long-short,
  list/display = used
}

\DeclareAcronym{SIL}{
  short = SIL,
  long = Safety Integrity Level
}
\DeclareAcronym{GaN HEMT}{
  short = GaN HEMT,
  long = Gallium Nitride High-Electron-Mobility Transistors
}
\DeclareAcronym{CCPS}{
  short = CCPS,
  long = Cryogenic Cyber Physical Systems
}
\DeclareAcronym{CPS}{
  short = CPS,
  long = Cyber Physical System
}
\DeclareAcronym{LSN}{
  short = LSN,
  long = Logical Synchrony Networks
}
\DeclareAcronym{HTS}{
  short = HTS,
  long = High Temperature Superconducting
}
\DeclareAcronym{SSTL}{
  short = SSTL,
  long = Synchronous Signal Temporal Logic
}
\DeclareAcronym{STL}{
  short = STL,
  long = Signal Temporal Logic
}
\DeclareAcronym{EMC}{
  short = EMC,
  long = Electromagnetic Compatibility
}
\DeclareAcronym{FPGA}{
  short = FPGA,
  long = Field Programmable Gate Array
}
\DeclareAcronym{FSM}{
  short = FSM,
  long = Finite State Machine
}
\DeclareAcronym{FOC}{
  short = FOC,
  long = Field Oriented Control
}
\DeclareAcronym{ADC}{
  short = ADC,
  long = Analogue to Digital Converter
}
\DeclareAcronym{IoT}{
  short = IoT,
  long = Internet of Things
}
\DeclareAcronym{IP}{
  short = IP,
  long = Intellectual Property
}
\DeclareAcronym{JSON}{
  short = JSON,
  long = JavaScript Object Notation
}
\DeclareAcronym{LTL}{
  short = LTL,
  long = Linear Temporal Logic
}
\DeclareAcronym{ML}{
  short = ML,
  long = Machine Learning
}
\DeclareAcronym{OS}{
  short = OS,
  long = Operating System
}
\DeclareAcronym{PWM}{
  short = PWM,
  long = Pulse Width Modulation
}
\DeclareAcronym{SNR}{
  short = SNR,
  long = Signal-to-Noise Ratio
}
\DeclareAcronym{LET}{
  short = LET,
  long = Logical Execution Time
}
\DeclareAcronym{pLTL}{
  short = pLTL,
  long = Probabilistic Linear Temporal Logic
}
\DeclareAcronym{HIL}{
  short = HIL,
  long = Hardware-in-the-Loop
}
\DeclareAcronym{TTP}{
  short = TTP,
  long = Time Triggered Protocol
}
\DeclareAcronym{VaV}{
  short = VaV,
  long = Verification and Validation
}
\DeclareAcronym{CPS-Week}{
  short = CPS-Week,
  long = Cyber-Physical Systems Week
}
\DeclareAcronym{SIH}{
  short = SIH,
  long = Signal Invariance Hypothesis
}
\DeclareAcronym{SDTA}{
  short = SDTA,
  long = Synchronous Discrete Timed Automata
}

\newcommand{\ignore}[1]{{}}

\title{Synchronous Signal Temporal Logic for Decidable Verification of Cyber-Physical Systems}

\author{Partha Roop}
\affiliation{
  \institution{University of Auckland}
  \city{Auckland}
  \country{New Zealand}
}
\email{p.roop@auckland.ac.nz}

\author{Sobhan Chatterjee}
\affiliation{
  \institution{University of Auckland}
  \city{Auckland}
  \country{New Zealand}
}
\email{sobhan.chatterjee@auckland.ac.nz}

\author{Avinash Malik}
\affiliation{
  \institution{University of Auckland}
  \city{Auckland}
  \country{New Zealand}
}
\email{avinash.malik@auckland.ac.nz}

\author{Nathan Allen}
\affiliation{
  \institution{Auckland University of Technology}
  \city{Auckland}
  \country{New Zealand}
}
\email{nathan.allen@aut.ac.nz}

\author{Logan Kenwright}
\affiliation{
  \institution{University of Auckland}
  \city{Auckland}
  \country{New Zealand}
}
\email{logan.kenwright@auckland.ac.nz}

\begin{abstract}

  Many \ac{CPS} work in a safety-critical environment, where correct
 execution, reliability and trustworthiness are essential. \ac{STL}
 provides a formal framework for checking safety-critical \acs{CPS}.
 However, static verification of \acs{STL} is undecidable,in general, except when
 we want to verify using run-time-based methods, which have
 limitations. We propose \ac{SSTL}, a decidable fragment of \acs{STL},
 which admits static safety and liveness property verification. In
  \acs{SSTL}, we assume that a signal is sampled at fixed discrete steps,
 called \emph{ticks}, and then propose a hypothesis, called the
  \ac{SIH}, which is inspired by a similar hypothesis for synchronous programs. We
 define the syntax and semantics of \acs{SSTL} and show that \acs{SIH}
 is a necessary and sufficient condition for equivalence between an
  \acs{STL} formula and its \acs{SSTL} counterpart. By translating
  \acs{SSTL} to LTL$_{\mathcal{P}}$ (LTL defined over predicates), we enable decidable model checking using the SPIN
 model checker. We demonstrate the approach on a 33-node human heart
 model and other case studies.

\end{abstract}

\keywords{SSTL, temporal logic, static verification, model checking, liveness properties, safety properties}

\begin{document}

\maketitle

\section{Introduction}
\label{sec:introduction}

Autonomous systems such as agricultural robots, self-driving vehicles,
manufacturing units, amongst others, are a subclass of
\ac{CPS}~\cite{alur2015principles}, where (usually) discrete and
distributed controllers control a continuous plant. Such systems operate
within environments with human presence. Hence, guaranteeing their safe
operation is essential. Many different techniques for autonomous safety
have been proposed. These techniques can be
classified into \textcircled{1} static
verification~\cite{althoff2021set}, which makes sure that the system is
safe for operation before deployment. However, it is well known that
static verification of \ac{CPS} is undecidable even in simple
cases~\cite{henzinger1995s}. \textcircled{2} Runtime verification of
\ac{CPS}~\cite{dreossi2019compositional,
  dalrymple2024guaranteedsafeaiframework}, where the verification engine
checks the required properties on the deployed system. Runtime
verification is complementary to static verification, since it does not
enumerate all possible execution paths of the system. \textcircled{3}
Controller synthesis~\cite{belta2019formal}, which produces a
controller for a given plant while satisfying the
formal properties. 
The controller synthesis problem, which is formulated as a constrained optimisation problem, has scalability issues for
large systems.

Irrespective of the chosen option, formal properties of \ac{CPS} are
specified in a suitable temporal logic with predicates over continuous variables.
\ac{STL}~\cite{donze_on_2013_stl} is the most widely used temporal logic
for property specification of autonomous CPS. Here, static~\cite{lercher2024using} and
runtime verification~\cite{yu_stlmc_2022,bae_bounded_2019} and
controller
synthesis~\cite{yang_continuous-time_2020,lindemann_control_2019} have been studied.
Static verification of \ac{STL} properties is
undecidable~\cite{bae2019bounded}. Overapproximating the system
trajectories can mitigate this problem, to some extent. However, such
over-approximation is only applicable for safety properties (\emph{nothing bad ever
happens}). To the best of our knowledge, \emph{liveness} (that \emph{something
good eventually happens}) properties, expressed in \ac{STL}, have never been statically verified
for autonomous systems. We propose to address the challenges of
undecidability of \ac{STL} property verification for both safety and
liveness in this work. These methods are essential for the verification of \emph{environment constraints} 
for safe autonomy to be realisable.

We focus on
environmental inputs and plant outputs, sensed as time series data, also
known as \emph{signals}. For example, consider a 33-node heart model scenario where we focus on two key nodes in the heart conduction system: the sinoatrial (SA) node and the atrioventricular (AV) node. At any instant in time, the state of the heart can be represented by a vector containing the electrical potentials at these two nodes. Formally, the state vector $\vec{x}$ is a map, $\vec{x}: \mathbb{R}_{\geq0} \rightarrow \mathbb{R}^{2}$, where the domain captures any instant in time and the range gives the pair of potentials at the SA and AV nodes at that time. Here, the progression of time is continuous, whose domain is $\mathbb{R}_{\geq0}$, also known as \emph{real-time systems}. 
Verification of signals in the real-time setting is undecidable. Here, we consider an alternative,
also known as the \emph{discrete time} setting, where time progresses in discrete steps, 
whose domain is $\mathbb{N}_{\geq 0}$. This approach is used in many practical systems, including 
all systems which are realised on a discrete computer, say by sampling some input.
A programming paradigm, called \emph{synchronous programming}~\cite{benveniste_synchronous_2003},
is based on such environmental sampling using discrete instants called \emph{ticks}.
The rationale for using the synchronous formulation for discritising \ac{STL} is as follows:

\begin{enumerate}
\item Any signal from a physical source has inertia proportional to its mass. For example, the ECG sensor. Here, after a \texttt{P}, which means an atrial depolarisation (beat), there is a gap before the next event \texttt{Q}, which indicates the start of the ventricular depolarisation.
\item In signal processing, in addition to transforms over continuous
 signals, there are transforms over the discrete versions. Consider the
 Fourier transform (FT) and the discrete Fourier transform (DFT). While
 FT is suitable for a theoretical analysis over the continuous time
 signal, DFT is suitable for the design of computational algorithms
 over sampled signals. Hence, we argue that a synchronous variant of
  \ac{STL} is needed, where we discretise the logic itself.
\item In synchronous programming \cite{benveniste_synchronous_2003},
 which is widely used for \ac{CPS}, a hypothesis is used to constrain
 the environment. This is known as the \emph{synchrony hypothesis},
 which states that \emph{the control system reacts infinitely fast
 relative to its plant}. In other words, it requires that data from
 sensors does not arrive faster than the rate at which the controller
 can process the data. This hypothesis can be adapted in our setting to
 formalise the relationship between \ac{STL} and its abstraction, which
 we term \emph{\ac{SSTL}}, especially to formalise the status of a signal
 between two sample points.
\end{enumerate}

Based on the above, our premise is that a discrete variant of \ac{STL},
which discretises the \ac{STL} logic and is developed for
synchronous systems, called \ac{SSTL}, will lead to methods for scalable
verification of \ac{CPS}. We also argue that, under suitable assumptions
of \emph{robustness}~\cite{fainekos_robustness_2009}, \ac{SSTL} is a
sound and complete abstraction of \ac{STL}. The main contributions of
the paper are as follows:

\begin{enumerate}
\item We propose the first synchronous (discrete) abstraction of the
  \ac{STL} logic, called \ac{SSTL}. We introduce a decidable encoding of
  \acf{STL} called \acf{SSTL}. 
\item We define the syntax and semantics of \ac{SSTL} and prove that
 under a suitable assumption of signal characteristics, this can be
 used as a sound and complete abstraction of \ac{STL}.
\item We present an LTL$_{\mathcal{P}}$ translation of \ac{SSTL} and use it to model check \ac{SSTL} properties using the SPIN model checker. This makes the problem of model checking \ac{SSTL} properties decidable and practical.
\item Another key contribution is the formalisation of a sound human heart model using formal temporal logic and its verification using SPIN.
\end{enumerate}

The rest of the paper is organised as follows. In
Section~\ref{sec:motivating-example}, we present a motivating example of
a heart model to illustrate the challenges and requirements for
discrete temporal logic specifications. In
Section~\ref{sec:methodology}, we formally define \ac{SSTL}, including
its syntax, semantics, and the \acf{SIH} that enables sound and complete
abstraction from \ac{STL}. Section~\ref{sec:verification-sstl-properties} presents our
verification of \ac{SSTL} properties. Section~\ref{sec:evaluation} presents our
evaluation of \ac{SSTL} properties. Section~\ref{sec:related-work} discusses related work
in temporal logic variants and verification methods for \ac{CPS}.
Finally, Section~\ref{sec:conclusion} summarises our contributions and
discusses future research directions.

\section{Motivating Example}
\label{sec:motivating-example}

In this section, we present a motivating case study based on a 33-node human heart model, which is representative of digital healthcare systems where safety and liveness properties are critical. This case study highlights the limitations of continuous-time temporal logic specifications, such as \ac{STL}, and motivates the need for a discrete, synchronous variant suitable for digital implementations.

\subsection{Heart Model Case Study}
\label{sec:heart-model-scenario}

A well-tested and verified heart model is crucial for the development of safe and reliable digital healthcare systems, such as pacemakers and defibrillators. Often, a heart model that accurately mimics the human heart and any diseases that may be present in the heart is required. However, the complexity of the heart model and its electrical properties can make it challenging to verify properties using traditional model-checking techniques. Here, we utilise a 33-node heart model to illustrate the challenges of verifying properties using \ac{STL} over continuous-time signals and demonstrate how our SSTL framework can be employed to verify properties for the heart model.

\begin{figure}[htbp]
  \centering
  \includegraphics[width=\linewidth]{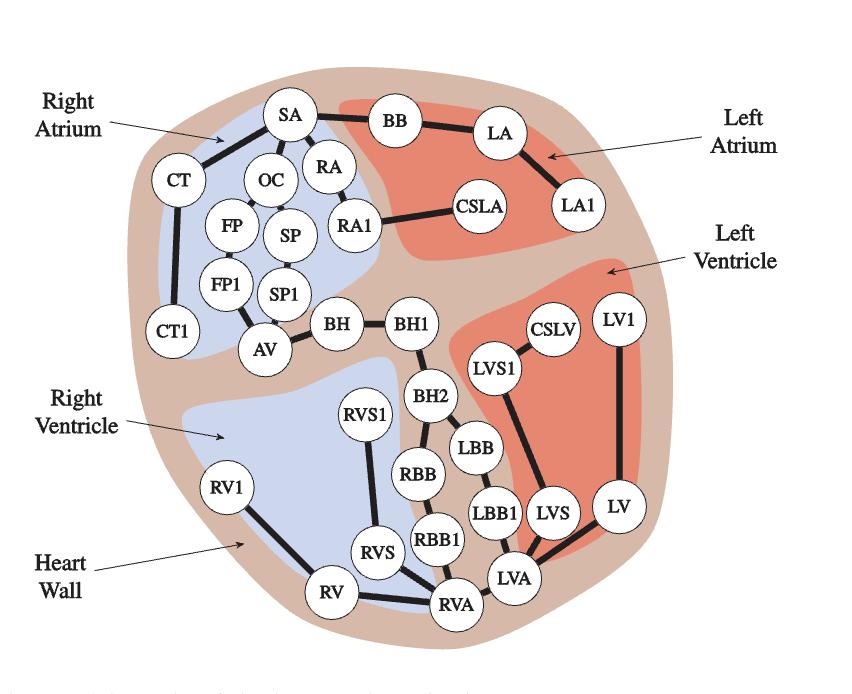}
  \caption{The 33-node heart model. (Adapted from \cite{yip2018towards})}
  \label{fig:heart_model}
\end{figure}

The 33-node heart model~\cite{yip2018towards} represents the heart's electrical conduction via interconnected nodes for atrial, atrioventricular, and ventricular regions (see Figure~\ref{fig:heart_model}). Electrical impulses travel from the sinoatrial (SA) node through to the atrioventricular (AV) node, then along bundle branches to the ventricles. The main signals of interest are the atrial and ventricular electrograms (A\_EGM and V\_EGM), plus each node's action potential, all tracked over time. 

To specify and verify timing requirements over these signals, we use \ac{STL}, which allows formal properties such as demanding ventricular activation (V\_EGM) follows atrial activation (A\_EGM) within a given window, ensuring physiologically safe heart rhythms.

\subsection{STL Specification for Heart Model}
\label{sec:stl-spec-heart}

To demonstrate the application of \ac{STL} to the heart model, we design four critical policies that capture essential safety and performance requirements for cardiac electrical activity. These specifications illustrate the types of temporal properties that must be verified in digital healthcare systems.

\paragraph{Property 1: Atrioventricular Conduction Safety}
\label{sec:policy-av-conduction}

The first policy ensures proper atrioventricular conduction timing, which is critical for maintaining coordinated heart rhythm:

\begin{align}
\varphi_{AV} := \square \left( A_{EGM} > V_{a,th} \rightarrow \lozenge_{[0.180,0.240]}(V_{EGM} > V_{v,th}) \right) 
\label{eq:av_conduction}
\end{align}

This specification requires that whenever atrial activation occurs ($A_{EGM}$ exceeds threshold $V_{a,th}$), ventricular activation must follow within 180ms to 240ms ($V_{EGM}$ exceeds threshold $V_{v,th}$). This property ensures that the AV node conducts electrical impulses from the atria to the ventricles properly, preventing dangerous delays.

\paragraph{Property 2: Ventricular Refractory Period}
\label{sec:policy-ventricular-refractory}

The second policy enforces the ventricular refractory period, preventing premature ventricular contractions that could lead to arrhythmias:

\begin{align}
\varphi_{refractory} := \square \left( V_{EGM} > V_{v,th} \rightarrow \lozenge_{[0.600,1.00]}(V_{EGM} \leq V_{v,th}) \right) 
\label{eq:ventricular_refractory}
\end{align}

This specification ensures that after each ventricular activation, the ventricular myocardium enters the refractory region ($V_{EGM}$ below threshold) within 600ms to 1000ms after the ventricular activation. This property prevents the occurrence of dangerous rapid ventricular rhythms and ensures adequate time for ventricular filling.

\paragraph{Property 3: Liveness Property}
\label{sec:policy-liveness}
The fourth property ensures that the heart model never deadlocks. It will eventually issue an atrial activation and a ventricular activation.

\begin{align}
\varphi_{liveness\_A} &:=  \lozenge (A_{EGM} > V_{a,th})\\
 \varphi_{liveness\_V} &:= \lozenge (V_{EGM} > V_{v,th})
\label{eq:liveness}
\end{align}

\begin{example}
 Consider the following \ac{STL} formula specifying an atrioventricular conduction property:
  \begin{equation}
  \varphi_{AV} = \square\left( A_{EGM} > V_{a,th} \rightarrow \lozenge_{[0.180,0.240]} V_{EGM} > V_{v,th} \right)
  \end{equation}
 This formula says: ``\emph{whenever there is an atrial activation}'' ($A_{EGM} > V_{a,th}$), ``\emph{a ventricular activation must occur within the next 180ms to 240ms}'' ($V_{EGM} > V_{v,th}$).
  
 Now, consider a continuous signal where $A_{EGM}$ and $V_{EGM}$ are defined for every time $t \in [0, \infty)$. To evaluate $\varphi_{AV}$ at $t=0.625$ seconds, where $A_{EGM}>V_{a,th}$ becomes true, we must verify that there is a later time (within 180ms to 240ms, i.e., [805, 865] milliseconds) where $V_{EGM} > V_{v,th}$. This procedure is illustrated in Figure \ref{fig:sstl_heart_example_stl}. Both $V_{a,th}$ and $V_{v,th}$ are assumed to be 80 mV. As we can see, atrial activation occurs at $t=0.625$ seconds, and ventricular activation occurs around 835 milliseconds, within 180ms to 240ms. Hence, the \ac{STL} formula is satisfied at $t=0.625$ seconds.
\end{example}

\begin{figure*}[htbp]
  \centering
  \begin{subfigure}{0.48\textwidth}
    \centering
    \includegraphics[width=\linewidth]{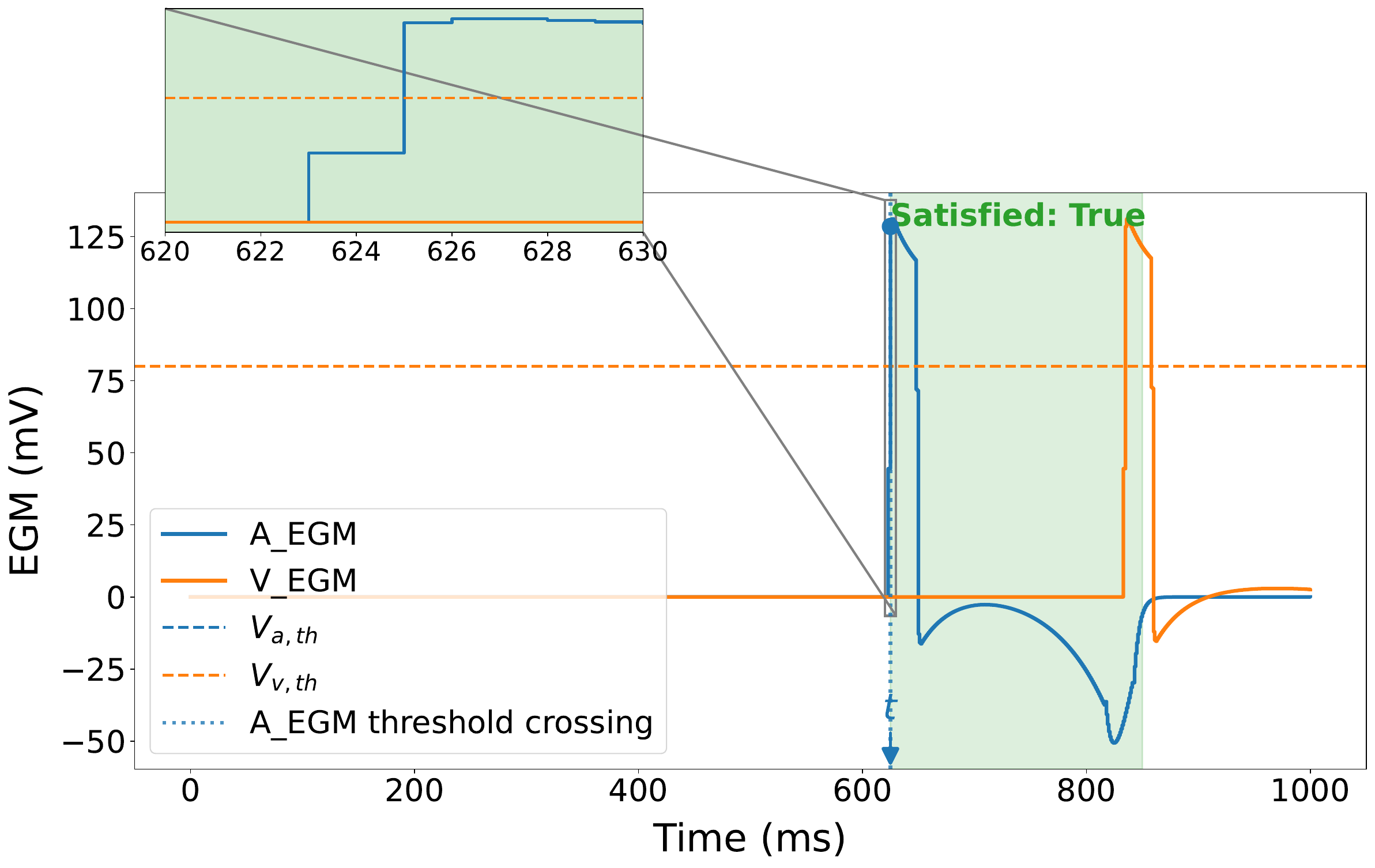}
    \caption{Evaluation of STL formula ($\varphi = \square (A_{EGM} > V_{a,th} \rightarrow \lozenge_{[0.180,0.240]} (V_{EGM} > V_{v,th}))$) at $t = 0.625$ seconds.}
    \label{fig:sstl_heart_example_stl}
  \end{subfigure}
  \hfill
  \begin{subfigure}{0.48\textwidth}
    \centering
    \includegraphics[width=\linewidth]{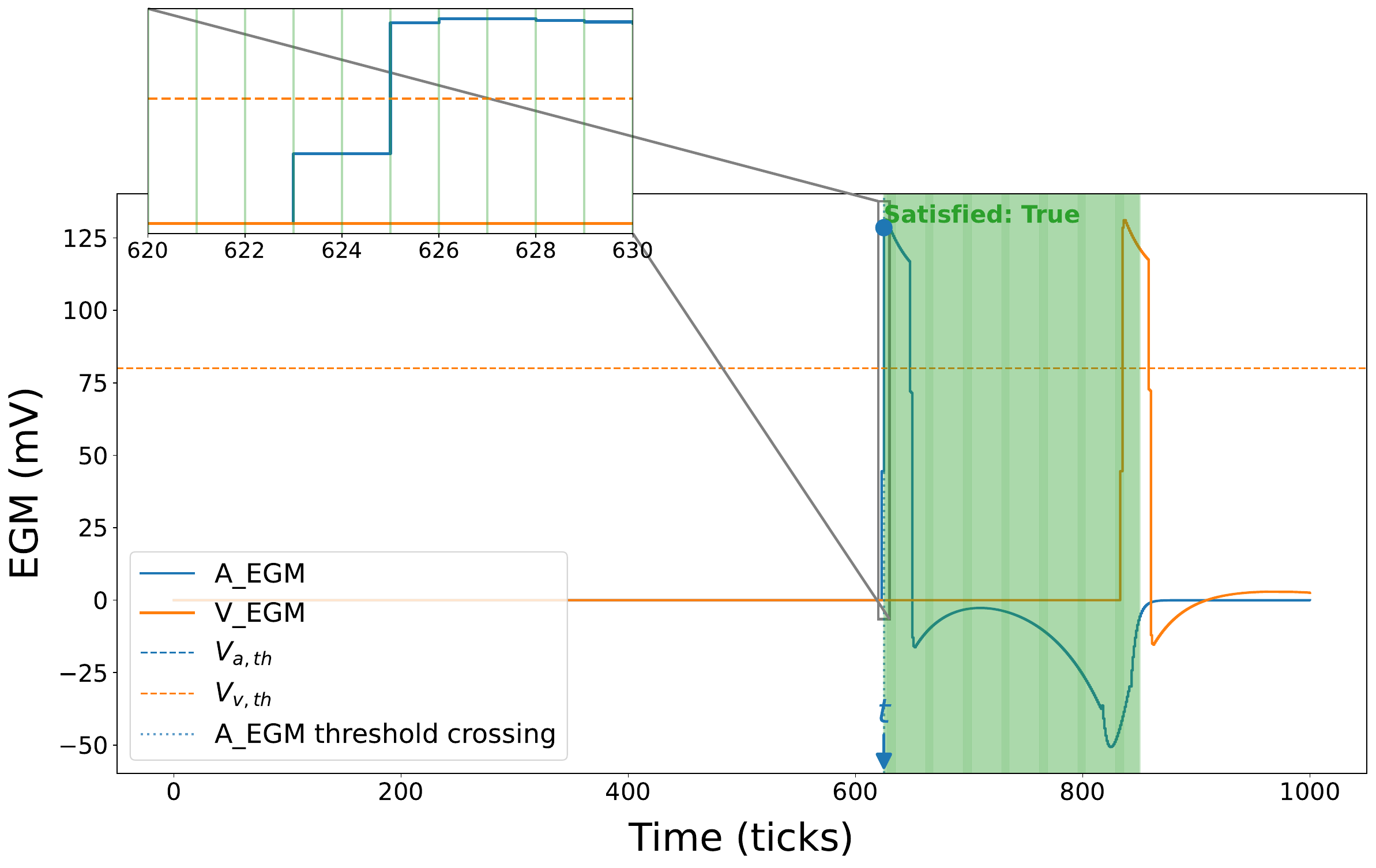}
    \caption{Evaluation of SSTL formula ($[\varphi] = \square (A_{EGM} > V_{a,th} \rightarrow \lozenge_{[[0.180],[0.240]]} (V_{EGM} > V_{v,th}))$) at $t = 625$ ticks.}
    \label{fig:sstl_heart_example_sstl}
  \end{subfigure}
  \caption{Comparison of STL and SSTL for the heart model scenario.}
  \label{fig:sstl_heart_example}
\end{figure*}

\subsection{Challenges with \ac{STL}}
\label{sec:challenges-stl-digital-healthcare}
While \ac{STL} offers a formal language for specifying temporal properties, it faces several practical limitations for digital healthcare systems.

\textbf{Continuous vs. Discrete Time}: \ac{STL} is defined over continuous time, whereas in practice, physiological signals are sampled at discrete intervals.

\textbf{Intractable Verification}: Checking properties like Equation~\ref{eq:av_conduction} in \ac{STL} requires evaluating all points in a real interval (see Figure~\ref{fig:sstl_heart_example_stl}), making general verification undecidable and computationally infeasible~\cite{bae2019bounded}.

\textbf{Sampling Issues}: Discrete sampling introduces challenges:
\emph{(i) Missing Critical Events:} Events violating safety properties may occur between samples and go undetected. 
\emph{(ii) Ambiguity in Temporal Operators:} The meaning of intervals like "within 80ms" ($\lozenge_{[0,0.08]}$) depends on the sampling rate and may lead to inconsistent evaluation.

\section{Synchronous Signal Temporal Logic (SSTL)}
\label{sec:methodology}

In this section, we develop SSTL, a synchronous variant of STL
designed specifically for CPS that operate under discrete
time. Our approach is to extend STL for signals obtained through
synchronous execution of physical systems, which we call SSTL. In
SSTL, time progresses in discrete time instants or ticks,
and traces are obtained at tick boundaries. Since most physical signals
are continuous and real-valued over $\mathbb{R}_{\geq0}$, to validate using
SSTL, we need to project the real-time variable $t$ to a discrete time variable
$[t]$. This projection, combined with our SIH, ensures that we can
reason about discrete-time behaviour while maintaining soundness with
respect to the underlying continuous system.

\begin{definition}
    \label{def:projection_operation}
 (Time Discretisation)
 Let $t \in \mathbb{R}_{\geq0}$ be a real-time variable and $\Delta t \in \mathbb{R}_{\geq0}$ represents the tick period/length in real-time. Then, we can define $[t]: \mathbb{R}_{\geq0} \rightarrow \mathbb{N}_{\geq0}$ as follows:
 \vspace{0.5em}

 For any $k\geq 0$ and $k \in \mathbb{N}_{\geq0}$, we have, $t \in [k\Delta t, (k+1)\Delta t) \implies [t] = k = \lfloor t/\Delta t \rfloor$. Therefore, for any interval $[a, b] \in \mathbb{R}_{\geq0}$, we have the discrete time interval, $[[a], [b]] \in \mathbb{N}_{\geq0}$, where $[a] = \lfloor a/\Delta t \rfloor$ and $[b] = \lfloor b/\Delta t \rfloor$.
    
    $\square$
\end{definition}

\begin{definition}
  \label{def:sstl_trace}
 (SSTL Trace)
 An SSTL trace $w_d$ is a function $w_d: \mathbb{N}_{\geq0} \rightarrow \mathbb{R}^n$ that maps discrete time points to signal values. For a given trace $w_d$ and discrete time $[t] \in \mathbb{N}_{\geq0}$, $x^{w_d}_i([t])$ denotes the value of signal $x_i$ at time $[t]$ in trace $w_d$, for $i = 1, 2, \ldots, n$.
  $\square$
\end{definition}

\begin{definition}
  \label{def:sih}
 (Signal Invariance Hypothesis (SIH))
The result of the time discretisation operation is that the signal remains invariant between two consecutive discrete ticks. Specifically, the value of a signal does not change within any single tick interval. Formally, for all $t, t' \in [k\Delta t, (k+1)\Delta t)$, where $t, t' \in \mathbb{R}_{\geq 0}$ and $k \in \mathbb{N}_{\geq 0}$, the following holds:
\begin{equation}
 x_i(t) = x_i(t') \quad \text{and} \quad [t] = k.
\end{equation}
Therefore, for any $t \in [k\Delta t, (k+1)\Delta t)$, we have
\begin{equation}
 x_i^{w_d}([t]) = x_i^{w}(t).
\end{equation}

The SIH is crucial for soundness, as it guarantees that no information is lost between discrete ticks due to changes in the underlying continuous signal.
$\square$
\end{definition}

\begin{remark}
 A stricter form of the SIH requires that the sampling frequency of the input signal (from sensors or a continuous process) satisfy the Nyquist criterion, so that the discrete trace is a sound abstraction of the continuous signal. Mathematically, $F_s \geq \kappa F_m$, where $F_s = 1/\Delta t$ is the sampling frequency, $F_m$ is the highest frequency component (or bandwidth) of the input signal, and $\kappa \in \mathbb{N}_{\geq2}$. Sampling at or above this rate preserves the signal's key features (no aliasing) and ensures that no relevant change in the continuous signal occurs \emph{between} two consecutive samples that would be missed by evaluating predicates only at tick instants. Thus, the piecewise-constant representation (constant over each tick interval) does not lose information needed for satisfaction of SSTL formulae, and the SIH is upheld. Maintaining the Nyquist rate is therefore essential for sound discrete-time verification of continuous-time properties.

 The interval boundaries in SSTL properties can be aligned with continuous-time signal boundaries by choosing $\kappa$ appropriately, so that sampling instants coincide with the natural timing of the signal and the discrete representation remains consistent with the continuous one. For a \emph{discrete} process that emits a single value at each tick (e.g., a clock or event generator with period $T$), the signal is defined only at those instants; choosing $\Delta t$ equal to or a multiple of $T$ ensures that the process is sampled at its natural rate and the SIH is satisfied by construction, since there is no underlying continuous signal between ticks. $\square$
\end{remark}

\begin{definition}
  \label{def:sstl_formula}
  
 If $\varphi$ is any STL formula, then the
 corresponding SSTL formula $[\varphi]$ is defined recursively as follows:

  \begin{align} [\varphi] := \top\ |\  [x_{i}^{w_d}(t) \geq 0]\ |\  \neg[\varphi]\ |\  [\varphi_{1}] \land [\varphi_{2}]\ |\ 
 [\varphi_{1}]\mathcal{U}_{[[a], [b]]}[\varphi_{2}]
  \end{align}

 where, $a \in \mathbb{R}_{\geq0}, \text{ and } b \in \mathbb{R}_{\geq0}$.
  $\square$
\end{definition}

\begin{definition}
  \label{def:SSTL_boolean_semantics}
 (SSTL Boolean Semantics)
 For an SSTL trace $w_d$ and discrete time $[t]$, the boolean semantics is defined recursively as:
    \begin{align}
 (w_d,[t]) &\models \top \nonumber \\
 (w_d,[t]) &\models [x_i^{w_d}(t) \geq 0] &&\iff x^{w_d}_i([t]) \geq 0 \nonumber \\
 (w_d,[t]) &\models \neg[\varphi] &&\iff (w_d,[t]) \not\models [\varphi] \nonumber \\
 (w_d,[t]) &\models [\varphi_1] \land [\varphi_2] 
      &&\iff (w_d,[t]) \models [\varphi_1] \land (w_d,[t]) \models [\varphi_2] \nonumber 
 \end{align}
\begin{align}
  (w_d,[t]) \models [\varphi_1] \mathcal{U}_{[[a], [b]]} [\varphi_2]
  \iff\; \exists\; [t'] \in [[t+a], [t+b]]\; :~ \nonumber \\
    (w_d,[t']) \models [\varphi_2] \nonumber \\
    \text{and}~ \forall\; [t''] \in [[t], [t']): (w_d,[t'']) \models [\varphi_1]
\end{align}
 We can redefine other usual operators as syntactic abbreviations:
  \begin{align}
 [\varphi_1] \vee [\varphi_2] &:= \neg(\neg[\varphi_1] \land \neg[\varphi_2]) \\
 [\varphi_1] \implies [\varphi_2] &:= \neg[\varphi_1] \lor [\varphi_2] \\
    \lozenge_{[[a], [b]]} [\varphi] &:= \top \mathcal{U}_{[[a], [b]]} [\varphi] \label{eq:exist-to-until}\\
    \square_{[[a], [b]]}[\varphi] &:= \neg\lozenge_{[[a], [b]]} \neg[\varphi]
  \end{align}
\end{definition}

Untimed operators like $\mathcal U$, $\lozenge$ and $\square$ are just abbreviations of the bounded operators with the interval $[0, +\infty)$. Thus, they have the same semantics as the bounded operators.

\begin{theorem}
  \label{theorem:stl_to_sstl_satisfaction}
 If a trace $w$\footnote{Please see Appendix \ref{appendix:syntax_and_semantics_of_stl} for the definition of STL trace, syntax and semantics.} contains signals that satisfy the SIH and the trace also satisfies the STL formula $\varphi$, then for any $t \in \mathbb{R}_{\geq0}$ and $[t] \in \mathbb{N}_{\geq0}$, we have:
  \[
 (w,t) \models \varphi \iff (w_d,[t]) \models [\varphi]
  \]
 The proof of this theorem is provided in the Appendix \ref{appendix:stl_to_sstl_satisfaction}.
\end{theorem}

\begin{example}
 Consider the heart model property presented in Equation \ref{eq:av_conduction}. Here, both the $A_{EGM}$ and $V_{EGM}$ signals satisfy SIH as they are produced and sampled at a fixed frequency of $F_s = 1/\Delta t = 1000$ Hz. Figure \ref{fig:sstl_heart_example} shows the difference in evaluation between STL and SSTL for this property. Suppose we are evaluating the satisfaction of the STL formula at $t_0 = 0.625$ seconds (shown by the blue dotted line), where $A_{EGM}>V_{a,th}$ becomes true. For STL, we must check if $A_{EGM}(t_0) > V_{a,th}$ at time $t_0$, then there must exist $t' \in [t_0+0.180, t_0+0.240]$ such that $V_{EGM}(t') > V_{v,th}$. The satisfaction check involves \emph{every} point in this real-valued interval. We can see that the STL formula is satisfied at $t_0 = 0.625$ seconds.

 For the SSTL case, assuming the signal is sampled at a fixed interval $\Delta t = 0.001$ seconds (e.g., 1 ms per tick), and we are evaluating at tick $[t_0] = \lfloor 0.625 / 0.001 \rfloor = 625$. The satisfaction check of the SSTL formula
  \[
 [\varphi_{AV}] := \square \left( A_{EGM} > V_{a,th} \rightarrow \lozenge_{[[0.180],[0.240]]} (V_{EGM} > V_{v,th}) \right)
  \]
requires projecting the formula time interval to the discrete domain. We need to check, for each tick in the interval $[625 + \lfloor 0.180/0.001 \rfloor, 625 + \lfloor 0.240/0.001 \rfloor]$ $= [625 + 180, 625 + 240] = [805, 865]$, whether $V_{EGM}$ exceeds $V_{v,th}$, provided $A_{EGM}$ exceeds $V_{a,th}$ at $[t_0]$. We can see that the SSTL formula is satisfied at $t_0 = 0.625$ seconds or $[t_0] = 625$ ticks.

 If the SIH property holds (i.e., the signal is invariant between two consecutive discrete ticks), then the satisfaction of the SSTL formula is equivalent to the satisfaction of the STL formula. This is because the satisfaction check at each discrete tick matches the satisfaction check over the corresponding real-valued interval, and vice versa. Additionally, while the STL evaluation necessitates checking every point in the real-valued interval, the SSTL evaluation only requires checking every discrete tick (indicated by green vertical lines in the figure).

\end{example}

\section{Verification of SSTL properties}
\label{sec:verification-sstl-properties}

To enable automated verification of \ac{SSTL} formulae through model
checking, we present a translation of \ac{SSTL} to LTL$_{\mathcal{P}}$ and then use
the SPIN model checker to verify  over state-based and kinematic models. We also show how
the LTL$_{\mathcal{P}}$ translation is a sound and complete abstraction of
\ac{SSTL}.

\subsection{\acf{LTL} with predicates}
\label{sec:acfltl-with-pred}

To translate \ac{SSTL} formulae to a form suitable for model checking,
we use \acf{LTL} with predicates (LTL$_{\mathcal{P}}$), taking inspiration from
the work presented in~\cite{kwon2008hybrid}, which extends traditional
LTL by integrating predicates defined over real-valued variables. The
LTL$_{\mathcal{P}}$ logic allows us to express temporal properties while
maintaining predicates over continuous or discrete-valued signals.

\begin{definition}
  \label{def:ltlp_syntax}
 (LTL$_{\mathcal{P}}$ Syntax) The syntax of LTL$_{\mathcal{P}}$ formulas $\phi$ is defined as follows:
  \begin{align}
    \phi &::=\top \mid P \mid \neg \phi \mid \phi_1 \land \phi_2 \mid \phi_1 \lor \phi_2 \mid \bigcirc \phi \mid \phi_1 \mathcal{U} \phi_2 \\
 P &::= a_1 y_1 + a_2 y_2 + \ldots + a_n y_n \bowtie b
  \end{align}
  
  \noindent
 where $P: \mathbb{R}^n \rightarrow \mathbb{B}$ is a predicate symbol representing a relational condition over variables $y_1, y_2, \ldots, y_n$. Specifically, $P$ takes the form $a_1 y_1 + a_2 y_2 + \ldots + a_n y_n \bowtie b$, where $a_1, \ldots, a_n$ are real-valued coefficients, $b$ is an offset, and $\bowtie$ is a relational operator chosen from $\{<, \leq, =, \geq, >\}$.

 The temporal operators are:
  \begin{itemize}
  \item $\bigcirc \phi$ (next): $\phi$ holds in the next state
  \item $\phi_1 \mathcal{U} \phi_2$ (until): $\phi_1$ holds until $\phi_2$ becomes true
  \end{itemize}

 The operators $\Box$ and $\Diamond$ can be defined as syntactic abbreviations:
  \begin{align}
    \Diamond \phi &:= \top \mathcal{U} \phi \\
    \Box \phi &:= \neg \Diamond \neg \phi
  \end{align}
\end{definition}

\begin{definition}
  \label{def:ltlp_semantics}
 (LTL$_{\mathcal{P}}$ Semantics): An LTL$_{\mathcal{P}}$ trace $\sigma$ is an infinite sequence
 of states, where each state $s_j, \forall j \in \mathbb{N}_{\geq0}$ assigns values to variables
  $y_1, \ldots, y_n$. For a state $s_{j}$, we write
  $s_{j}(y_1), \ldots, s_{j}(y_n)$ to denote the values of
 variables $y_1, y_2, \ldots, y_n$ in state $s_{i}$. The satisfaction
 relation $(\sigma, j) \models \phi$ (read as ``$\sigma$ satisfies
  $\phi$ at position $j$'') is defined recursively as follows:
  
  \begin{align}
 (\sigma, j) \models P &\iff a_1 s_{j}(y_1) + a_2 s_{j}(y_2) + \ldots + a_n s_{j}(y_n) \bowtie b \\
 (\sigma, j) \models \neg \phi &\iff (\sigma, j) \not\models \phi \\
 (\sigma, j) \models \phi_1 \land \phi_2 &\iff (\sigma, j) \models \phi_1 \text{ and } (\sigma, j) \models \phi_2 \\
 \end{align}
 \begin{align}
 (\sigma, j) \models \phi_1 \lor \phi_2 &\iff (\sigma, j) \models \phi_1 \text{ or } (\sigma, j) \models \phi_2 \\
 (\sigma, j) \models \bigcirc \phi &\iff (\sigma, j+1) \models \phi \\
 (\sigma, j) \models \phi_1 \mathcal{U} \phi_2 &\iff \exists k \geq j \text{ such that } (\sigma, k) \models \phi_2 \\
               &\quad \text{ and } \forall l \in [j, k), \sigma, l \models \phi_1 \\
 (\sigma, j) \models \Box \phi &\iff \forall k \geq j, (\sigma, k) \models \phi \\
 (\sigma, j) \models \Diamond \phi &\iff \exists k \geq j, (\sigma, k) \models \phi
  \end{align}
  
 We say that a sequence $\sigma$ satisfies an LTL$_{\mathcal{P}}$ formula
  $\phi$ (denoted $\sigma \models \phi$) if $\sigma, 0 \models \phi$.
\end{definition}

\begin{example}
 Consider the following LTL$_{\mathcal{P}}$ formula involving two variables $y_1$ and $y_2$ (where $y_1$ is temperature $x$ in Celsius and $y_2$ is voltage $v$ in volts):
  \[
    \phi := 2y_1 - y_2 \geq 10
  \]
 That is, $\phi$ is true in a state if $2$ times the temperature minus the voltage is at least $10$.

 Now, suppose we have an LTL$_{\mathcal{P}}$ trace $\sigma = s_0, s_1, s_2, \ldots$ where each state assigns values to these variables.
  
 Take as an example the state $s_5$ with $s_5(y_1) = 15$ and $s_5(y_2) = 18$. Substituting these values into the formula, we have:
  \[
    2 \cdot s_5(y_1) - s_5(y_2) = 2 \cdot 15 - 18 = 30 - 18 = 12
  \]
 Since $12 \geq 10$, the formula $\phi$ holds at state $s_5$, i.e., $(\sigma, 5) \models \phi$.
\end{example}

\subsection{SSTL to LTL translation }
\label{sec:sstl-ltl-translation}

\begin{figure}[htbp]
  \centering
  \includegraphics[width=\linewidth]{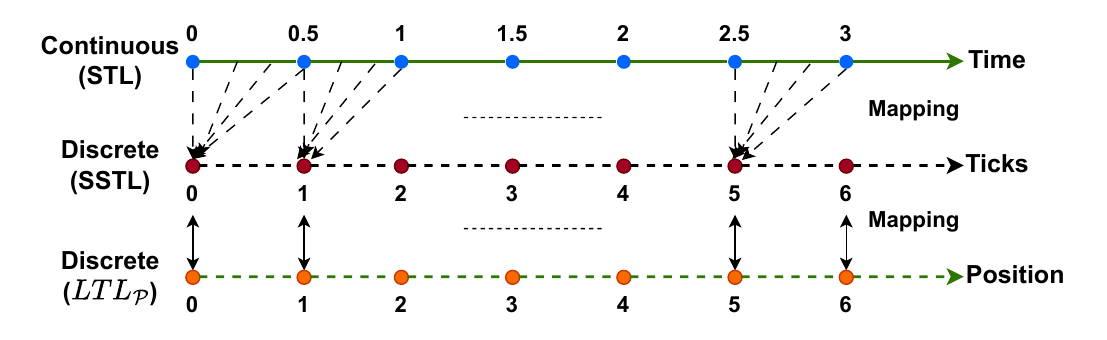}
  \caption{Illustration of the mapping between the continuous time and the discrete time assuming that the signal is sampled at a fixed interval $\Delta t = 0.5$ seconds}
  \label{fig:time_mapping}
\end{figure}

\begin{definition}
  \label{def:ltlp_sstl_time_bijection}
 (SSTL to LTL$_{\mathcal{P}}$ time bijection)
 We formalise the correspondence between positions in the SSTL trace and the LTL$_{\mathcal{P}}$ trace as follows. Let $w_d = (x_0^{w_d}, x_1^{w_d}, \ldots)$ be a discrete-time signal trace where $x_k^{w_d}$ is the value at discrete tick $k \in \mathbb{N}_{\geq0}$. Let $\sigma = (s_0, s_1, \ldots)$ be the corresponding LTL$_{\mathcal{P}}$ trace, where each $s_j$ gives the state at position $j$. By Definition~\ref{def:projection_operation}, we define a mapping $\pi : \mathbb{N}_{\geq 0} \to \mathbb{N}_{\geq 0}$ such that
  \begin{equation}
 j = \pi([t]) := [t]
  \end{equation}
 for each tick $[t]$. That is, the $x_{[t]}^{w_d}$-th value in $w_d$ corresponds to the $j$-th position in $\sigma$, yielding a bijection between the two traces at each tick. Formally,
  \[
    \forall [t] \in \mathbb{N}_{\geq 0}, \quad x[t]^{w_d} \leftrightarrow s_j
  \]

 Thus, the satisfaction of the SSTL formula $[\varphi]$ at time $[t]$ is equivalent to the satisfaction of the SSTL formula $[\varphi]$ at position $j$.

  $\square$
\end{definition}

\begin{definition}
    \label{def:ltlp_sstl_function_translation}
 (SSTL to LTL$_{\mathcal{P}}$ function)
 Let $w_d$ be an SSTL trace and let
    $\sigma$ be the same trace viewed as an LTL$_{\mathcal{P}}$ trace. Let the set of SSTL formulae be denoted by $\llbracket\varphi\rrbracket$ and that of LTL$_{\mathcal{P}}$ formula be denoted by $\llbracket\phi\rrbracket$. We write $\tau: \llbracket\varphi\rrbracket \rightarrow \llbracket\phi\rrbracket$ for the LTL$_{\mathcal{P}}$ formula $\phi$ obtained by translating an SSTL formula $[\varphi]$.
    \begin{equation}
      \phi = \tau([\varphi])
    \end{equation}
  $\square$
\end{definition}

\begin{definition}
  \label{def:ltlp_sstl_function_translation_unbounded}
 (SSTL to LTL$_{\mathcal{P}}$ function translation for unbounded intervals)
The translation for SSTL syntax cases can be divided into two parts - one that deals with formulae with unbounded intervals ($[0, \infty]$) and one that deals with formulae with bounded intervals ($[a, b]$). Given Definition \ref{def:ltlp_sstl_time_bijection}, there is a one-to-one correspondence between the SSTL trace and the LTL$_{\mathcal{P}}$ trace. The translation of the unbounded case using the function $\tau$ is straightforward and is as follows:

\begin{align}
  \tau\bigl(x_i^{w_d}([t])\ge0\bigr) &:= s_{j}(x_i^{w_d}) \ge 0 \\
  \tau(\neg[\varphi]) &:= \neg\tau([\varphi]),\\
  \tau([\varphi_1]\land[\varphi_2]) &:= \tau([\varphi_1])\land\tau([\varphi_2]),\\
  \tau\bigl([\varphi_1]\;\mathcal U\;[\varphi_2]\bigr)
  &:= \tau([\varphi_1])\;\mathcal U\;\tau([\varphi_2])\\
  \tau(\square[\varphi]) &:= \square\tau([\varphi])\\
  \tau(\lozenge[\varphi]) &:= \lozenge\tau([\varphi])
\end{align}

\end{definition}

To handle bounded time intervals in the LTL$_{\mathcal{P}}$ formula, we introduce a position index $j$ that denotes the current trace position (i.e., the current discrete tick). The index $j$ increases by one at every step. When a time-bounded obligation begins (for example, when evaluating $\Box(A\,\rightarrow\,\Diamond_{[a,b]}V)$ at a tick where $A$ holds, or at the entry tick of a bounded-Until), we capture the entry position as $j_0$ and keep it fixed for the duration of that particular obligation. If multiple obligations overlap in time, each has its own (conceptual) copy of $j_0$. With this, any interval $[a,b]$ with $a,b\in\mathbb{N}_{\ge0}$ can be expressed as a predicate over $j$ and $j_0$. We define the guard:

\begin{equation}
   \textsf{within}[a,b]\;:=\; (j_0+a \le j \le j_0+b),
\end{equation}
which is a plain LTL$_{\mathcal P}$ atomic predicate evaluated over the current state $s_j$ \footnote{Please check Appendix~\ref{appendix:implementation_details} for more details.}. 

\begin{lemma}
\label{lemma:boundedness-j0}
(Boundedness of the Number of $j_0$ Copies)\\
Let an SSTL formula be interpreted over a (possibly infinite) discrete trace, using its translation to LTL$_{\mathcal{P}}$ with time-bounded operators. Then, for any position $j$, the number of distinct active obligations (i.e., distinct conceptual $j_0$ copies that must be tracked at $j$) depends only on the number of active obligations at $j$, as determined by the structure of the formula and trace up to that point.

The proof sketch is given in Appendix~\ref{appendix:boundedness_j0}.
$\square$
\end{lemma}

Using this idea, the bounded operators are derived as follows when evaluated at tick $[t]$ and position $j$. The other operators are translated in the same way as the unbounded case.

\begin{align}
  \tau\bigl([\varphi_1]\;\mathcal U_{[[a],[b]]}\;[\varphi_2]\bigr) &:= \tau([\varphi_1])\;\mathcal U\;\bigl(\tau([\varphi_2]) \land \textsf{within}[a,b]\bigr),\\
  \tau\bigl(\lozenge_{[[a],[b]]}[\varphi]\bigr) &:= \Diamond\bigl(\tau([\varphi]) \land \textsf{within}[a,b]\bigr),\\
  \tau\bigl(\square_{[[a],[b]]}[\varphi]\bigr) &:= \Box\bigl( \,\textsf{within}[a,b] \,\rightarrow \, \tau([\varphi])\bigr)
\end{align}

\begin{example}
 We illustrate the translation from SSTL to LTL$_{\mathcal{P}}$ for the Until operator, both in unbounded and time-bounded cases.

  \textbf{Formulas:}
  
  \begin{itemize}
    \item \emph{Unbounded Until:}
      \[
 [\varphi] := x_1^{w_d}([t]) \ge 0 \;\mathcal{U}\; x_2^{w_d}([t]) \ge 0
      \]
    \item \emph{Bounded Until (with interval $[5,10]$):}
      \[
 [\varphi] := x_1^{w_d}([t]) \ge 0 \;\mathcal{U}_{[[5],[10]]}\; x_2^{w_d}([t]) \ge 0
      \]
  \end{itemize}

  \textbf{Signals Table:}
  
 The concrete values of $x_1^{w_d}$ and $x_2^{w_d}$ over positions $j = [t]$ are as follows:

  \begin{center}
    \begin{tabular}{c|ccccccccccc}
      $[t]$ & 0 & 1 & 2 & 3 & 4 & 5 & 6 & 7 & 8 & 9 & 10 \\
      \hline
      $x_1^{w_d}$ & 1 & 1 & 1 & 0.5 & 0.8 & 0.2 & 1 & 0.5 & 0.2 & -1 & -0.7 \\
      $x_2^{w_d}$ & -1 & -1 & -0.8 & -0.6 & -0.5 & -0.1 & -0.15 & 0.6 & 1 & 1 & 0.8 \\
    \end{tabular}
  \end{center}

  \textbf{Translations:}
  
  \begin{itemize}
    \item \emph{Unbounded Until:}
      \begin{align*}
        \tau([\varphi]) &= \tau\bigl(x_1^{w_d}([t]) \ge 0\bigr) \;\mathcal{U}\; \tau\bigl(x_2^{w_d}([t]) \ge 0\bigr) \\
        &= (s_j(x_1^{w_d}) \ge 0) \;\mathcal{U}\; (s_j(x_2^{w_d}) \ge 0)
      \end{align*}
    \item \emph{Bounded Until:}
      \begin{align*}
        \tau([\varphi]) &= \tau\bigl(x_1^{w_d}([t]) \ge 0\bigr) \;\mathcal{U}\; \bigl(\tau\bigl(x_2^{w_d}([t]) \ge 0\bigr) \land \textsf{within}[5,10]\bigr) \\
        &= (s_j(x_1^{w_d}) \ge 0) \;\mathcal{U}\; \bigl((s_j(x_2^{w_d}) \ge 0) \land \textsf{within}[5,10]\bigr)
      \end{align*}
 where $\textsf{within}[5,10] := (j_0+5 \le j \le j_0+10)$, ensuring that $x_2^{w_d} \ge 0$ must become true between ticks 5 and 10, counted from the start of the Until obligation.
  \end{itemize}

  \textbf{Truth Table and Explanation:}

 The following table gives, for each position $j$, the truth values for all subformulas, the time window predicate, and both the unbounded and bounded Until formulas.

  \begin{center}
    \renewcommand{\arraystretch}{1.2}
    \begin{tabular}{p{3cm}|ccccccccccc}
      $j = [t]$ & 0 & 1 & 2 & 3 & 4 & 5 & 6 & 7 & 8 & 9 & 10 \\
      \hline
      \hline
      $\varphi_1 = x_1^{w_d}[t]\geq0$ & 1 & 1 & 1 & 1 & 1 & 1 & 1 & 1 & 1 & 0 & 0 \\
      $\varphi_2 = x_2^{w_d}[t]\geq0$ & 0 & 0 & 0 & 0 & 0 & 0 & 0 & 1 & 1 & 1 & 1 \\
      \hline
      $\varphi_1\ \mathcal U\ \varphi_2$ & 1 & 1 & 1 & 1 & 1 & 1 & 1 & 1 & 1 & 0 & 0 \\
      \hline
      $\textsf{within}[5,10]$ & 0 & 0 & 0 & 0 & 0 & 1 & 1 & 1 & 1 & 1 & 1 \\
      $\tau(\varphi_2) \land \textsf{within}[5,10]$ & 0 & 0 & 0 & 0 & 0 & 0 & 0 & 1 & 1 & 1 & 1 \\
      $\tau(\varphi_1)\;\mathcal U\;(\tau(\varphi_2) \land \textsf{within}[5,10])$ & 0 & 0 & 1 & 1 & 1 & 1 & 0 & 0 & 0 & 0 & 0 \\
      \hline
      $\varphi_1\ \mathcal U_{[5,10]}\ \varphi_2$ & 0 & 0 & 1 & 1 & 1 & 1 & 0 & 0 & 0 & 0 & 0 \\
    \end{tabular}
  \end{center}

  \emph{Explanation}: 
  \begin{itemize}
    \item The first two rows record where $\varphi_1$ ($x_1^{w_d}[t]\geq0$) and $\varphi_2$ ($x_2^{w_d}[t]\geq0$) are true.
    \item The third row computes the unbounded Until over the trace: as soon as $\varphi_1$ breaks at $j=9$, the Until becomes false.
    \item The predicate $\textsf{within}[5,10]$ is true only from ticks 5 to 10, marking the time window for the bounded Until.
    \item $\tau(\varphi_2) \land \textsf{within}[5,10]$ shows at which ticks both the second signal is non-negative and we are within the window. 
    \item The row $\tau(\varphi_1)\;\mathcal U\;(\tau(\varphi_2) \land \textsf{within}[5,10])$ gives the satisfaction of the bounded-Until LTL$_{\mathcal{P}}$ translation.
    \item The final row ($\varphi_1\ \mathcal U_{[5,10]}\ \varphi_2$) matches the translation, showing that the SSTL semantics and LTL$_{\mathcal{P}}$ translation agree.
  \end{itemize}
 The bounded temporal operator enforces that $x_2^{w_d}\ge0$ must hold between ticks $j_0+5$ and $j_0+10$, with $x_1^{w_d}\ge0$ true at each prior tick, otherwise the bounded Until fails.
\end{example}

\subsubsection{Nesting of Formulas}
\label{sec:nesting-formulas}

The translation function $\tau$ is applied recursively to handle nested formulas. When a formula contains subformulas, each subformula is translated independently, and the results are combined according to the structure of the parent formula. This recursive approach ensures that complex nested temporal formulas are correctly translated to LTL$_{\mathcal{P}}$.

For example, consider a nested formula where a bounded operator contains another temporal operator:
\[
[\varphi] := \square_{[[0],[20]]} \bigl(x_1^{w_d}([t]) \ge 0 \;\mathcal{U}\; x_2^{w_d}([t]) \ge 0\bigr)
\]

The translation proceeds recursively:
\begin{align*}
  \tau([\phi]) &= \tau\bigl(\square_{[[0],[20]]} \bigl(x_1^{w_d}([t]) \ge 0 \;\mathcal{U}\; x_2^{w_d}([t]) \ge 0\bigr)\bigr) \\
  &= \Box\bigl(\textsf{within}[0,20] \rightarrow \tau\bigl(x_1^{w_d}([t]) \ge 0 \;\mathcal{U}\; x_2^{w_d}([t]) \ge 0\bigr)\bigr) \\
  &= \Box\bigl(\textsf{within}[0,20] \rightarrow \bigl(\tau(x_1^{w_d}([t]) \ge 0) \;\mathcal{U}\; \tau(x_2^{w_d}([t]) \ge 0)\bigr)\bigr) \\
  &= \Box\bigl(\textsf{within}[0,20] \rightarrow \bigl((s_j(x_1^{w_d}) \ge 0) \;\mathcal{U}\; (s_j(x_2^{w_d}) \ge 0)\bigr)\bigr)
\end{align*}

When multiple bounded operators are nested, each introduces its own time constraint with a separate $j_0$ value. For instance, in a formula like $\lozenge_{[[a],[b]]} \square_{[[c],[d]]} \phi$, the outer $\lozenge$ operator captures $j_0$ when it becomes active, and the inner $\square$ operator captures its own $j_0'$ when it becomes active within the outer operator's scope. This ensures that each bounded temporal operator correctly tracks its own time interval relative to its activation point.

\begin{theorem}
  \label{theorem:sstl_to_ltlp_correctness}
 If we have a SSTL trace $w_d$ and an LTL$_{\mathcal{P}}$ trace $\sigma$ such that the SSTL formula $[\varphi]$ is translated to the LTL$_{\mathcal{P}}$ formula $\tau([\varphi])$, using the translation function $\tau$ from Definition \ref{def:ltlp_sstl_function_translation}, then:
  \[
 (w_d,[t]) \models [\varphi] \quad\iff \quad (\sigma,j) \models \tau([\varphi])
  \]

 The proof of this theorem is provided in the Appendix \ref{appendix:sstl_to_ltlp_correctness}.
\end{theorem}

\begin{lemma}[Decidability of SSTL Verification]
  \label{remark:decidability}
 The verification problem for SSTL is decidable. This follows directly from Theorem~\ref{theorem:sstl_to_ltlp_correctness} and the decidability of LTL model checking. Since every SSTL formula $[\varphi]$ can be translated to an equivalent LTL$_{\mathcal{P}}$ formula $\tau([\varphi])$, and since LTL$_{\mathcal{P}}$ model checking over finite-state systems is decidable~\cite{holzmann_spin_2011}, SSTL verification is also decidable.
  
\end{lemma}

\begin{proof}[Proof sketch]
 The result follows by reduction to LTL$_{\mathcal{P}}$ model checking via
 Theorem~\ref{theorem:sstl_to_ltlp_correctness}.  For every SSTL formula
  $[\varphi]$ we can construct an equivalent LTL$_{\mathcal{P}}$ formula
  $\tau([\varphi])$.  Model checking LTL over finite-state systems is
 decidable: given a transition system~$M$ and an LTL$_{\mathcal{P}}$ formula~$\phi$, one
 builds a B\"uchi automaton $\mathcal{A}_{\neg\phi}$ for the negation of the
 specification and checks whether $L(M) \cap L(\mathcal{A}_{\neg\phi}) =
  \emptyset$; the corresponding emptiness test is decidable and can be
 performed in time linear in the size of the product automaton
  \cite{holzmann_spin_2011}.

 Two additional arguments ensure that the reduction indeed yields a
  \emph{finite} product automaton:
  \begin{enumerate}
    \item LTL$_{\mathcal{P}}$ extends LTL with arithmetic
 predicates over program variables. In our setting, every variable
 ranges over a fixed, finite domain (bounded integers or
 fixed--precision reals), so the alphabet of valuations remains
 finite, and the automata produced by the translation are
 finite--state.
    \item The SSTL~$\rightarrow$~LTL$_{\mathcal{P}}$ translation introduces
 guards of the form $\textsf{within}[a,b] := (j_0+a \le j \le j_0+b)$
 that rely on a bookkeeping index~$j_0$. By
 Lemma~\ref{lemma:boundedness-j0}, at any position~$j$ only finitely
 many distinct $j_0$ copies are active, bounded by the structure of
 the formula and the widths of its time intervals. Hence, the state
 space required to evaluate such guards is also finite.
  \end{enumerate}

 With these observations, the model-checking product remains finite-state,
 and the decidability of SSTL verification thereof.\qedhere
\end{proof}

SSTL translation uniformly handles both safety and liveness properties:
\begin{itemize}
  \item \textbf{Safety properties} (``something bad never happens'') are expressed using $\Box$ and negation, and correspond to B\"uchi automata that reject traces violating the safety condition.
  \item \textbf{Liveness properties} (``something good eventually happens'') are expressed using $\Diamond$, and correspond to B\"uchi automata with acceptance conditions ensuring eventual satisfaction.
\end{itemize}

Furthermore, the complexity of SSTL model checking inherits the complexity bounds of LTL model checking: PSPACE-complete in the size of the formula and polynomial in the size of the model. In practice, modern model checkers like SPIN employ sophisticated optimisation techniques (partial-order reduction, symmetry reduction, state compression) that make verification tractable for realistic system models.

\begin{table*}[htbp]
  \centering
  \caption{SSTL properties verified across three case studies. Properties were verified using SPIN with verification times shown.}
  \label{tab:sstl-properties-all-cases}
  \begin{tabular}{llp{9cm}cc}
    \toprule
    \textbf{System} & \textbf{Property} & \textbf{SSTL Formula} & \textbf{Result} & \textbf{Time} \\
    \midrule
    \multirow{7}{*}{\begin{tabular}[c]{@{}l@{}}Traffic\\Light\end{tabular}} 
    & $[\varphi_{mutex}]$ & $\square \neg(NS_{green} = 1 \land EW_{green} = 1)$ & Satisfied & 1.713s \\
    & $[\varphi_{NS\_safe}]$ & $\square(NS_{green} = 1 \rightarrow EW_{red} = 1)$ & Satisfied & 1.502s \\
    & $[\varphi_{EW\_safe}]$ & $\square(EW_{green} = 1 \rightarrow NS_{red} = 1)$ & Satisfied & 1.624s \\
    & $[\varphi_{fairness}]$ & $\square \lozenge (NS_{green} = 1)$ & Not Satisfied & 1.559s \\
    & $[\varphi_{liveness\_NS}]$ & $\lozenge (NS_{green} = 1)$ & Satisfied & 1.475s \\
    & $[\varphi_{response\_NS\_yellow}]$ & $\square(NS_{green} = 1 \rightarrow \lozenge (NS_{yellow} = 1))$ & Satisfied & 1.539s \\
    & $[\varphi_{bounded\_NS\_yellow}]$ & $\square(NS_{green} = 1 \rightarrow \lozenge_{[[3],[5]]} (NS_{yellow} = 1))$ & Satisfied & 2.093s \\
    \midrule
    \multirow{4}{*}{\begin{tabular}[c]{@{}l@{}}Pedestrian\\Crossing\end{tabular}} 
    & $[\varphi_{no\_conflict}]$ & $\square \neg(cars_{green} = 1 \land walk_{signal} = 1)$ & Satisfied & 275.288s \\
    & $[\varphi_{queue\_bounded}]$ & $\square(waiting_{peds} \leq 5 \land waiting_{peds} \geq 0)$ & Satisfied & 309.154s \\
    & $[\varphi_{threshold\_walk}]$ & $\square(waiting_{peds} \geq 2 \rightarrow \lozenge (walk_{signal} = 1))$ & Not Satisfied & 2.293s \\
    & $[\varphi_{bounded\_wait}]$ & $\square((cars_{green} = 1 \land waiting_{peds} \geq 2) \rightarrow \lozenge_{[[2],[5]]} (walk_{signal} = 1))$ & Not Satisfied & 1.627s \\
    \midrule
    \multirow{4}{*}{\begin{tabular}[c]{@{}l@{}}Healthy\\heart\\model\end{tabular}} 
    & $[\varphi_{AV}]$ & $\square(A_{EGM} \geq V_{a,th} \rightarrow \lozenge_{[[0.180],[0.240]]} (V_{EGM} \geq V_{v,th}))$ & Satisfied & 65.879s  \\
    & $[\varphi_{VV}]$ & $\square(V_{EGM} \geq V_{v,th} \rightarrow \lozenge_{[[0.6],[1.00]]} (V_{EGM} > V_{v,th}))$ & Satisfied &57.497s \\
    & $[\varphi_{liveness\_A}]$ &$\lozenge (A_{EGM} > V_{a,th})$ & Satisfied & 6.757s \\
    & $[\varphi_{liveness\_V}]$ &$\lozenge (V_{EGM} > V_{v,th})$ & Satisfied & 6.437s  \\
    \bottomrule
  \end{tabular}
\end{table*}

\subsection{SSTL verification procedure}

Given an SSTL formula $[\varphi]$, we apply the translation function $\tau$ recursively to obtain an equivalent LTL$_{\mathcal{P}}$ formula $\phi = \tau([\varphi])$ and then verify the property $\phi$ on a model of the system with the SPIN model checker \cite{holzmann_spin_2011}.

\subsubsection{SPIN support for LTL$_{\mathcal{P}}$}

A crucial distinction between standard LTL and LTL$_{\mathcal{P}}$ is that
LTL$_{\mathcal{P}}$ atomic propositions can be arbitrary predicates over program
state variables, including arithmetic and relational expressions. While
classical LTL is defined over propositional atoms from a fixed set $AP$,
LTL$_{\mathcal{P}}$ allows predicates of the form $s_j(x_i) \geq 0$,
$(j_0 + a \leq j \leq j_0 + b)$, and compound expressions involving arithmetic
operations over integer and floating-point variables.

SPIN natively supports this extended form of LTL by allowing any valid
Promela (SPIN's programming language) boolean expression as an atomic proposition within an
LTL formula. When SPIN evaluates an LTL formula during model
checking, it evaluates each predicate in the context of the current
global state, where all Promela variables have concrete values. For
instance, the predicate $(A_{EGM} > A_{EGM,th})$ evaluates
to true or false based on the values of $A_{EGM}$ and
$A_{EGM,th}$ in the current state. Similarly, timing
constraints such as those shown in Equation \ref{eq:av_conduction} are evaluated
directly as arithmetic expressions over state variables.

It is important to note that while SPIN evaluates predicates containing arithmetic expressions, the underlying model checker still operates over a finite state space. Therefore, all variables appearing in LTL predicates must have finite domains (e.g., bounded integers or floating-point values with finite precision). We use three decimal places for floating-point values and convert them to integers by multiplying by 1000 and rounding to the nearest integer. This simple approach ensures that the state space is finite. Under SIH, the discrete-time trace faithfully represents the continuous-time trace (no information loss between ticks), so the SSTL semantics and the STL--SSTL equivalence are \emph{exact} for the \emph{continuous}-time trace. The rounding step is an \emph{implementation} requirement for finite-state model checking. Verification in SPIN is therefore exact with respect to the \emph{quantised} trace (rounded to the chosen precision), not the original real-valued trace. Increasing the number of decimal places (e.g., scaling by 10000 instead of 1000) refines the approximation at the cost of a larger state space.

\subsubsection{System Modelling in Promela}

We model the system in Promela as a combination of a finite-state machine (FSM) encoding the system's control logic and C code blocks encapsulating the continuous dynamics. The primary time reference in the model is a discrete tick counter variable, incremented on each Promela execution step, ensuring all sampled variables and events align to this discrete timeline. FSM states represent operational phases (such as "wait for atrium" or "wait for ventricle" in heart models), and transitions are triggered based on signal values or the passage of specific ticks/times. 

Physical signal computations are implemented in embedded C code, if needed, which Promela integrates via $c\_decl$, $c\_state$, and $c\_code$ blocks. Promela synchronises with these C routines at each tick: the simulation advances, critical state variables (such as electrogram signals or event timestamps) are updated and made visible to property checks, and the FSM logic reacts to new values. Execution bounds, such as a maximum tick count or event threshold, are enforced to guarantee a finite state space for verification. Based on the LTL$_{\mathcal{P}}$ property, we also define and track $j_0$ as the tick counter variable that marks the start of a time-bounded obligation.

\subsubsection{Verification with SPIN}

SPIN performs verification by constructing the product (intersection) of the system model—represented as a finite-state machine (FSM) defined in Promela—and the Büchi automaton corresponding to the negation of the given LTL$_{\mathcal{P}}$ property. The workflow is as follows: the user encodes the system as a Promela FSM (including any embedded C code for dynamics), and specifies the LTL$_{\mathcal{P}}$ property using the \texttt{ltl} declaration. SPIN internally translates the negation of this LTL formula into a Büchi automaton. 

During verification, SPIN explores the product of the system FSM and this property automaton, searching for counterexamples—that is, executions in which the system behaviour intersects with the set of behaviours violating the property. This approach enables SPIN to identify infinite executions (cycles) in the product automaton that witness property violations (for example, infinite occurrence of error states, or failure to satisfy a liveness guarantee). 

All reachable behaviours within the finite state space are examined. If an intersection is discovered—i.e., if there exists a trace accepted by the product of the system model and the automaton for the negated property—SPIN produces a counterexample trace for user inspection. Otherwise, if no such intersection exists, the property is verified to hold for all possible executions of the model, thus ensuring both safety and liveness with respect to the system under analysis.

\section{Evaluation}
\label{sec:evaluation}

In this section, we evaluate the effectiveness of our SSTL framework through three case studies of increasing complexity: a traffic light controller, a pedestrian crossing system, and a heart model. Each case study demonstrates how SSTL can express and verify critical temporal properties in various systems, from basic safety and liveness constraints to complex physiological dynamics with integer-valued signals. All case studies execute synchronously.

\begin{table*}[ht]
  \centering
  \caption{Discretisation methods vs. our SSTL (concise view)}
  \begin{tabular}{llll}
    \toprule
    \textbf{Work} & \textbf{Time handling} & \textbf{Key assumption/tool} & \textbf{Limitation vs. SSTL} \\
    \midrule
Fainekos\&Pappas~\cite{fainekos_robustness_2009} & Discrete check of CTL/MTL & Robustness, bounded sampling & Needs robustness margins; extra assumptions \\
Deshmukh et al.~\cite{bartocci_specification-based_2018} & Piecewise-constant interp. & RoSI, partial traces & Interp. choice affects semantics \\
Chen\&Francis~\cite{chen_optimal_1995} & Uniform discrete-time STL & Fixed sampling over $\mathbb{N}_{\ge0}$ & Breaks explicit link to dense-time STL \\
Donz\'e et al.~\cite{donze_efficient_2013} & Mixed-time with interfaces & $@_{cd}$/$@_{dc}$ operators & Interface discipline adds complexity \\
Yang et al.~\cite{yang_continuous-time_2020} & Uniform discretisation for planning & Robust semantics, MIQP & Optimization-oriented; not equivalence \\
    \bottomrule
  \end{tabular}
  \label{tab:discretisation-landscape}
\end{table*}

Table~\ref{tab:sstl-properties-all-cases} presents the key SSTL properties verified for each case study. All properties were verified using the SPIN model checker with the translation procedure described in Section~\ref{sec:sstl-ltl-translation}. The translation is performed manually. The table shows the SSTL properties, the result of the verification, and the time taken to verify the property.

The verification results demonstrate SSTL's effectiveness across diverse application domains with increasing complexity. The traffic light controller is a finite state discrete system and demonstrates verification (of simple boolean signals) of both \emph{safety properties} (mutual exclusion, safety interlocks) that specify what must never occur, and \emph{liveness properties} that guarantee eventual progress. This system is sampled at the same rate as the heart model's data generation rate. Of the seven properties verified, six are satisfied, with the fairness property $[\varphi_{fairness}]$ not being satisfied, which indicates that there is a problem in the model or the specification.

The pedestrian crossing system is a discrete system with integer-valued signals and demonstrates SSTL's ability to handle integer-valued signals, verifying properties over both boolean control signals and integer queue sizes. This system is sampled at the same rate as the model execution rate. The system successfully verifies safety properties (no conflicting flows, queue boundedness) with verification times ranging from 2 to 309 seconds, demonstrating the framework's capability to reason about mixed signal types. We see from Table~\ref{tab:sstl-properties-all-cases} that all properties are satisfied, but the bounded response properties ($[\varphi_{threshold\_walk}]$ and $[\varphi_{bounded\_wait}]$), suggesting potential issues with the model or the specification. Additionally, it takes significantly less time (in seconds, as opposed to hundreds of seconds) to identify a fault in a faulty model, demonstrating the efficiency of the SSTL framework when our goal is to determine if the model is faulty.

The heart model is sampled at a fixed interval of 0.001 seconds and demonstrates SSTL's expressiveness for complex physiological systems with precise timing constraints. The heart model is implemented in C and then used in Promela with a top-level abstraction of the heart model to verify the properties. All four properties, which are expressed in continuous time (like in STL) with SSTL syntax, are satisfied \footnote{We run the verification to a depth of 200,000, in SPIN, due to memory constraints, and no errors were found.} by the healthy heart model, including time-bounded response properties $[\varphi_{AV}]$ and $[\varphi_{VV}]$ that capture critical atrioventricular (AV) conduction delays and ventricular refractory periods. Verification times range from 6.4-65.9 seconds, with the longer times for the bounded temporal properties reflecting the complexity of reasoning over integer-valued electrogram signals with millisecond-precision timing constraints.

We also tested the same heart model properties by introducing diseases in the model. We tested for AV Conduction Block, LBB Conduction Block, and RBB Conduction Block \footnote{These blocks prevent the conduction of the electrical signal from one section of the heart to another. Results are omitted here for brevity}. We find that the property $[\varphi_{AV}]$ is not satisfied by the diseased heart model in all three diseases, while the properties $[\varphi_{VV}]$, $[\varphi_{liveness\_A}]$ and $[\varphi_{liveness\_V}]$ are satisfied. Thus, we can infer that these diseases affect AV conduction delays but do not affect the liveness of the heart model.

However, we see higher memory requirements for complex systems with larger state spaces, like the pedestrian crossing system and the heart model, both of which have integer-valued (or converted to integer-valued) signals. The memory requirements reach 7000 MB and 9000 MB for the pedestrian crossing system and the heart model, respectively. Nevertheless, the memory requirements can be reduced using SPIN's memory reduction features (which may increase the verification time) and using model abstraction to reduce the state space. Other, more efficient LTL model checking techniques can also be used to reduce the memory requirements.

Overall, a key insight from these case studies is the generality of the SSTL framework: \emph{any property expressible in STL that satisfies the Signal Invariance Hypothesis (SIH) can be systematically converted to SSTL and subsequently translated to LTL$_{\mathcal{P}}$ for verification}. This is on top of SSTL's ability to directly verify discrete systems. The verified properties span the full spectrum of temporal logic expressiveness: safety properties using the always operator ($\square$), liveness properties using eventually ($\lozenge$), bounded response properties using time-constrained operators ($\lozenge_{[[a],[b]]}$, $\square_{[[a],[b]]}$), and complex combinations thereof. 
\section{Related Work}
\label{sec:related-work}

  \begin{table*}[ht]
    \centering
    \caption{STL verification approaches vs. our SSTL}
    \begin{tabular}{llll}
      \toprule
      \textbf{Work} & \textbf{Core idea} & \textbf{Strength} & \textbf{Limitation vs. SSTL} \\
      \midrule
 Roehm et al.~\cite{roehm_stl_2016} & Reach-set abstraction & Sound/complete for samples & Requires reachability; costly abstractions \\
 Bae\&Lee~\cite{bae_bounded_2019} & Syntactic separation + SMT & Refutational completeness (bounded) & Complex SMT over reals; bounded scope \\
 Yu et al.~\cite{yu_stlmc_2022} & $\epsilon$-strengthening + SMT & Robust model checking & Guarantees up to thresholds only \\
 Lercher\&Althoff~\cite{lercher2024using} & Four-valued STL on sets & Incremental verification & Depends on set refinement strategy \\
 Belta\&Sadraddini~\cite{belta2019formal} & DC programming (robustness trees) & Convex subproblems & Restrictive dynamics/spec classes \\
      \bottomrule
    \end{tabular}
    \label{tab:stl-verification}
  \end{table*}

\subsection{Discretisation Landscape}
\label{sec:stl-time-discretisation}
Table \ref{tab:discretisation-landscape} presents prior discretisation approaches. As shown, Fainekos and Pappas~\cite{fainekos_robustness_2009} propose testing continuous-time signals via discrete-time analysis using timed state sequences $\mu=(\sigma,\tau)$, where a sampling function $\tau$ and a bounding function $E(\cdot)$ enable interpolation. Their guarantees rely on variable or constant sampling with bounded step sizes and regularity conditions such as Lipschitz continuity or bounded derivatives. This yields a principled approach to monitoring continuous properties in discrete time, albeit at the cost of strong assumptions regarding signal smoothness. Deshmukh et al.~\cite{bartocci_specification-based_2018} instead advocate piecewise-constant interpolation and Robust Interval Semantics (RoSI) over discrete instants $\{t_0,\ldots,t_N\}$ without requiring uniform sampling, which is attractive for online, sliding-window monitoring; however, the piecewise-constant reconstruction can be too coarse for applications that need fidelity between samples.

Another line of work discretises the time itself. Chen et al.~\cite{chen_optimal_1995} define a discrete-time STL with uniform sampling over $\mathbb{N}{\ge 0}$ and add a cumulative-time operator $C^I\tau$, offering a clean framework but breaking the explicit tie to dense-time semantics crucial for physical reasoning. Ferrère et al.~\cite{donze_efficient_2013} develop Mixed-time STL with interface operators $@{cd}$ and $@{dc}$ to translate between continuous and discrete domains under periodic sampling (period $T$) and right-continuation, making time conversions explicit yet demanding careful interface discipline to preserve meaning. Pant et al.~\cite{yang_continuous-time_2020} adopt a uniform grid $[0\!:\!dt\!:\!T]$ and robust STL semantics for trajectory optimisation, assuming $dt\in\mathbb{R}^+$ and finite control input sequences; their focus is optimisation, not verification.

In contrast to these approaches, SSTL provides a new time-discretised temporal logic that maintains a proven equivalence to continuous-time STL semantics through the Signal Invariance Hypothesis, without requiring interpolation schemes or interface operators. This equivalence ensures that verification results obtained using SSTL directly correspond to properties of the underlying continuous system.

\subsection{STL Verification}
\label{sec:stl-verification}

Verifying STL over CPS is challenging due to the presence of continuous-time signals and, in general, is undecidable~\cite{bae_bounded_2019}. Table \ref{tab:stl-verification} summarises representative verification approaches by their core idea relative to our approach. As shown, Roehm et al.~\cite{roehm_stl_2016} abstract to reach sets, yielding sound/complete results for sampled-time STL at the cost of expensive reachability; Bae and Lee~\cite{bae_bounded_2019} use syntactic separation (STL-GT) to reduce bounded problems to SMT, but still require complex real-arithmetic encodings; and Yu et al.'s STLmc~\cite{yu_stlmc_2022} combines $\varepsilon$-strengthening with SMT to obtain robust guarantees up to a threshold, yet remains limited to quantifier-free fragments over reals.

Other directions similarly trade generality for tractability. As shown in the table \ref{tab:stl-verification}, Sato et al.~\cite{yang_continuous-time_2020} encode variable-interval STL as MILP, making decidability hinge on linear/rectangular dynamics; Lercher and Althoff~\cite{lercher2024using} adopt four-valued STL over reachable sets with incremental refinement, whose precision depends on the underlying reachability engine; and Takayama et al.~\cite{belta2019formal} address discrete-time STL via DC programming, achieving convex subproblems at the expense of restricting system and property classes.

Our SSTL framework provides a different approach: by operating directly in
the discrete time domain with proven equivalence to continuous-time STL, we
enable verification using discrete-time model checking techniques while
maintaining semantic soundness with respect to continuous-time properties.
This avoids the need for complex SMT encodings or reachability abstractions
while providing guarantees that hold for the underlying continuous system.

\section{Conclusions}
\label{sec:conclusion}

This paper introduces \ac{SSTL}, the first synchronous discrete abstraction of \ac{STL} designed for decidable verification of cyber-physical systems for both safety and liveness properties. We address the fundamental challenge of undecidability in STL verification by leveraging discrete time while maintaining soundness and completeness under the Signal Invariance Hypothesis (SIH). As long as the SIH property holds, STL properties can be translated to SSTL properties and verified using SPIN.

Our methodology consists of three key steps: first, we discretise continuous-time signals by applying the time projection operator $[t]: \mathbb{R}_{\geq0} \rightarrow \mathbb{N}_{\geq0}$, where $[t] = \lfloor t/\Delta t \rfloor$ maps real time to discrete ticks. The SIH ensures that signals remain invariant within each tick interval, preserving information between sample points. Second, we provide theoretical measures to calculate the sampling rate $\Delta t$ for a given application domain and the signal properties. Third, we translate SSTL formulas to LTL$_{\mathcal{P}}$ using a translation function $\tau$, where bounded temporal operators $\mathcal{U}_{[[a],[b]]}$, $\Box_{[[a],[b]]}$, $\lozenge_{[[a],[b]]}$ are encoded using the $\textsf{within}[a,b]$ predicate that captures time windows relative to obligation entry points. 

We verify the translated formulas using SPIN's exhaustive state-space exploration, leveraging its native support for arithmetic predicates over state variables. The verification process integrates continuous dynamics through embedded C code in Promela models, where physical simulations (e.g., ion channel dynamics in the cardiac model) execute at each state transition. SPIN evaluates predicates over real-valued signals (discretised to finite precision for a finite state space), boolean state indicators, and integer timing variables, providing on-the-fly evaluation during state exploration. The model checker constructs a Büchi automaton from the LTL$_{\mathcal{P}}$ formula and performs intersection with the system model's state graph, detecting violations through acceptance cycle detection.

Thus, SSTL provides the theoretical foundation and practical verification methodology to address this critical need, enabling the development of the next generation of provably safe autonomous cyber-physical systems. Future research directions include extending SSTL to include robustness measures, testing the framework on more complex models and trying different model checkers, integrating with runtime monitoring for complementary assurance, and exploring automatic property synthesis and property translation.

\begin{acks}
  The authors from the University of Auckland and Auckland University of Technology acknowledge the support of the Google grants
  Designing Scalable Synchronous Applications over Google bittide and Towards a formally
  verified autonomous vehicle by leveraging the bittide protocol. They also acknowledge the
  biweekly meetings with the Google bittide team, namely Calin Cascaval, Tammo Spalink, Sanjay Lall, and Martin Izzard. They would also like to thank Supratim Gupta for his valuable feedback and suggestions on the signal processing aspects of the paper.
\end{acks} 

\bibliographystyle{ACM-Reference-Format}
\bibliography{references}

\appendix
\section{Syntax and Semantics of STL}
\label{appendix:syntax_and_semantics_of_stl}
\begin{definition}
  \label{def:stl_trace}
  (STL Trace)
  A STL trace $w$ is a function $w: \mathbb{R}_{\geq0} \rightarrow \mathbb{R}^n$ that maps real-time to signal values. For a given trace $w$ and real-time $t \in \mathbb{R}_{\geq0}$, $x^{w}_i(t)$ denotes the value of signal $x_i$ at time $t$ in trace $w$, for $i = 1, 2, \ldots, n$.
\end{definition}

\begin{definition}
  \label{def:stl_syntax}
  (STL Syntax)
  The syntax of STL formulas $\varphi$ is defined as follows:
  \begin{align}
    \varphi &::=\top \mid x_i^w(t) \geq 0 \mid \neg \varphi \mid \varphi_1 \land \varphi_2 \mid \varphi_1 \mathcal{U}_{[a,b]} \varphi_2
  \end{align}
\end{definition}

\begin{definition}
  \label{def:stl_semantics}
  (STL Semantics)
  The semantics of STL formulas $\varphi$ is defined as follows:
  \begin{align}
    (w,t)\models\top \iff \top \text{ is true} \nonumber \\
    (w,t)\models x_i^{w}(t)\geq0 \iff x_i^{w}(t)\geq0 \nonumber \\
    (w,t)\models\neg\varphi \iff (w,t)\not\models\varphi \nonumber \\
    (w,t)\models\varphi_1\land\varphi_2 \iff (w,t)\models\varphi_1 \text{ and } (w,t)\models\varphi_2 \nonumber 
  \end{align}
  \begin{align}
    (w,t)\models\varphi_1\mathcal{U}_{[a,b]}\varphi_2 \iff \exists t_1\in[t+a,t+b]: (w,t_1)\models\varphi_2 \text{ and } \nonumber\\
     \forall t_2\in[t,t_1): (w,t_2)\models\varphi_1 \nonumber
  \end{align}
\end{definition}    

\section{Proof of Theorem \ref{theorem:stl_to_sstl_satisfaction}}
\label{appendix:stl_to_sstl_satisfaction}
\begin{proof}

\underline{Forward Direction:} Given $(w,t) \models \varphi$, we need to show that $(w_d,[t]) \models [\varphi]$. Let $t\in [k\Delta t, (k+1)\Delta t)$ for some $k \in \mathbb{N}_{\geq0}$. Then, $[t] = k$ and due to SIH, $\vec{x}(t) = \vec{x}([t])$. Proof is now based on the four cases of the STL formula $\varphi$ and the SSTL semantics.

\begin{itemize}
    \item \textbf{Case 1: $\varphi=\top$.} 
    
    Then $[\varphi]=\top$, and trivially
    \[
    (w,t)\models\top \implies (w_d,[t])\models\top.
    \]

    \item \textbf{Case 2: $\varphi=x_i^{w}(t)\geq0$.} 
    
    By definition $[\varphi]=x_i^{w_d}([t])\geq0$.
    By SIH we have $x_i^{w}(t)=x_i^{w_d}([t])$, hence
    \[
    (w,t)\models x_i^{w}(t)\geq0 \implies  (w_d,[t])\models x_i^{w_d}([t])\geq0.
    \]

\item \textbf{Case 3: $\varphi=\neg\psi$.} 

Then $[\varphi]=\neg[\psi]$.
\[
\begin{array}{r@{\ }c@{\ }l}
(w, t)\models\neg\psi
&\iff& (w, t)\not\models\psi\\
&\implies& (w_d,[t])\not\models[\psi]\quad(\text{SIH})\\
&\implies& (w_d,[t])\models\neg[\psi]\\
\therefore (w, t)\models\neg\psi &\implies& (w_d,[t])\models\neg[\psi].
\end{array}
\]

\item \textbf{Case 4: $\varphi=\varphi_1\land\varphi_2$.} 

Then $[\varphi]=[\varphi_1]\land[\varphi_2]$.
\[
\begin{array}{r@{\ }c@{\ }l}
\because (w, t)\models\varphi_1\land\varphi_2
&\iff& (w, t)\models\varphi_1 \ \text{ and } \ (w, t)\models\varphi_2\\
&\implies& (w_d,[t])\models[\varphi_1] \ \text{ and } \ (w_d,[t])\models[\varphi_2] \quad (\text{SIH})\\
&\implies& (w_d,[t])\models[\varphi_1]\land[\varphi_2]\\
\therefore (w, t)\models\varphi_1\land\varphi_2 &\implies& (w_d,[t])\models[\varphi_1]\land[\varphi_2].
\end{array}
\]

\item \textbf{Case 5: $\varphi=\varphi_1\ \mathcal{U}_{[a,b]}\ \varphi_2$.}

By definition $[\varphi]=[\varphi_1]\ \mathcal{U}_{[a,b]}\ [\varphi_2]$.

We must show
\[
(w, t)\models \varphi_1\ \mathcal{U}_{[a,b]}\ \varphi_2
\implies
(w,[t])\models [\varphi_1]\ \mathcal{U}_{[a,b]}\ [\varphi_2].
\]
Assume $(w, t)\models \varphi_1\ \mathcal{U}_{[a,b]}\ \varphi_2$. Then, by the dense-time STL semantics, there exists a real time $t_1\in[t+a,t+b]$ such that
\[
(w,t_1)\models\varphi_2
\quad\text{and}\quad
\forall t_2\in[t,t_1)\;:\;(w,t_2)\models\varphi_1.
\]

From the SSTl semantics of Until, $[t_1] \in [[t+a],[t+b]]$ and $[t_2] \in [[t],[t_1])$ and hence, by SIH we have, $x_i^{w}(t_1) = x_i^{w_d}([t_1])$ and $x_i^{w}(t_2) = x_i^{w_d}([t_2])$ and hence,
\[
(w_d,[t_1])\models[\varphi_2]
\quad\text{and}\quad
(w_d,[t_2])\models[\varphi_1].
\]
Consequently,
\[
(w_d,[t])\models [\varphi_1]\ \mathcal{U}_{[a,b]}\ [\varphi_2].
\]

$\therefore (w, t)\models \varphi_1\ \mathcal{U}_{[a,b]}\ \varphi_2 \implies (w_d,[t])\models [\varphi_1]\ \mathcal{U}_{[a,b]}\ [\varphi_2]$.

\end{itemize}

\underline{Backward Direction:} Given $(w_d,[t]) \models [\varphi]$, we must show that $(w, t) \models \varphi$. We prove this by contradiction, assuming that $(w, t) \not\models \varphi$ but $(w_d,[t]) \models [\varphi]$. Let $t \in [k\Delta t, (k+1)\Delta t)$ for some $k \in \mathbb{N}_{\geq0}$, so $[t] = k$. By the SIH, $x_i^{w}(t) = x_i^{w_d}([t])$.

\begin{itemize}
    \item \textbf{Case 1: $[\varphi]=\top$.} 
    
    Then $\varphi=\top$. The assumption $(w, t) \not\models \top$ is false, as $\top$ is always true.
    
    $\therefore (w_d,[t])\models\top \implies (w, t)\models\top$.

    \item \textbf{Case 2: $[\varphi]=x_i^{w_d}([t])\geq0$.}
    \begin{itemize}
        \item By assumption, $(w_d,[t]) \models x_i^{w_d}([t]) \geq 0$, which means $x_i^{w_d}([t]) \geq 0$.
        \item By SIH, $x_i^{w}(t) = x_i^{w_d}([t])$.
        \item Therefore, $x_i^{w}(t) \geq 0$, which means $(w,t) \models x_i^{w}(t) \geq 0$.
        \item This contradicts the assumption that $(w, t) \not\models \varphi$.
        
        $\therefore (w_d,[t])\models x_i^{w_d}([t])\geq0 \implies (w, t)\models x_i^{w}(t)\geq0$.
    \end{itemize}

    \item \textbf{Case 3: $[\varphi]=[\varphi_1]\land[\varphi_2]$.} 
    \begin{itemize}
        \item By SSTL semantics, $(w_d,[t])\models[\varphi_1]\land[\varphi_2]$ implies $(w_d,[t])\models[\varphi_1]$ and $(w_d,[t])\models[\varphi_2]$.
        \item By the SIH, since $(w_d,[t])\models[\varphi_1]$ and $(w_d,[t])\models[\varphi_2]$, it must be that $(w,t)\models\varphi_1$ and $(w,t)\models\varphi_2$.
        \item By STL semantics, $(w,t)\models\psi$ and $(w,t)\models\theta$ means $(w,t)\models\psi\land\theta$.
        \item This contradicts the assumption that $(w, t) \not\models \varphi_1\land\varphi_2$.
        
        $\therefore (w_d,[t])\models[\varphi_1]\land[\varphi_2] \implies (w, t)\models\varphi_1\land\varphi_2$.
    \end{itemize}

    \item \textbf{Case 4: $[\varphi]=\neg[\psi]$.}
    \begin{itemize}
        \item By assumption, $(w_d,[t])\models\neg[\psi]$ which, by SSTL semantics, implies $(w_d,[t])\not\models[\psi]$.
        \item By the SIH, if $(w_d,[t])\not\models[\psi]$, then it must be that $(w,t)\not\models\psi$.
        \item By STL semantics, $(w,t)\not\models\psi$ implies $(w,t)\models\neg\psi$.
        \item This contradicts the assumption that $(w, t) \not\models \neg\psi$.
        
        $\therefore (w_d,[t])\models\neg[\psi] \implies (w, t)\models\neg\psi$.
    \end{itemize}

    \item \textbf{Case 5: $[\varphi]=[\varphi_1]\mathcal{U}_{[a,b]}[\varphi_2]$.}
    \begin{itemize}
        \item By assumption, $(w_d,[t])\models[\varphi_1]\mathcal{U}_{[a,b]}[\varphi_2]$.
        \item By SSTL semantics, there exists a discrete time $[t'] \in [[t+a],[t+b]]$ such that $(w_d,[t']) \models [\varphi_2]$ and for all discrete times $[t''] \in [[t],[t'])$, $(w_d,[t'']) \models [\varphi_1]$.
        \item Let $t_{real} \in [t+a,t+b]$ be a real-time instant that maps to $[t']$, i.e., $[t_{real}] = [t']$.
        \item By the SIH, since $(w_d,[t']) \models [\varphi_2]$, it must be that $(w,t_{real}) \models \varphi_2$.
        \item Similarly, for all $t''_{real} \in [t,t_{real}]$ such that $[t''_{real}] \in [[t],[t'])$, since $(w_d,[t'']) \models [\varphi_1]$, it must be that $(w,t''_{real}) \models \varphi_1$.
        \item By the dense-time STL semantics, the existence of such a $t_{real}$ and a continuous interval of times in $[t,t_{real}]$ that satisfy the conditions means that $(w,t)\models\varphi_1\mathcal{U}_{[a,b]}\varphi_2$.
        \item This contradicts the assumption that $(w,t)\not\models\varphi_1\mathcal{U}_{[a,b]}\varphi_2$.
        
        $\therefore (w_d,[t])\models[\varphi_1]\mathcal{U}_{[a,b]}[\varphi_2] \implies (w, t)\models\varphi_1\mathcal{U}_{[a,b]}\varphi_2$.
    \end{itemize}
\end{itemize}

As both directions are proven, we have shown that $(w,t) \models \varphi \iff (w_d,[t]) \models [\varphi]$.

\end{proof}

%-------------------------------------------
\section{Proof of Theorem \ref{theorem:sstl_to_ltlp_correctness}}
\label{appendix:sstl_to_ltlp_correctness}
\begin{proof}

Let $w_d$ be an SSTL trace and let $\sigma$ be the same sequence of states viewed as an LTL$_{\mathcal P}$ trace. By Definition~\ref{def:ltlp_sstl_time_bijection}, the discrete position $j$ in $\sigma$ coincides with the logical tick $[t]$ of $w_d$, that is, $j=[t]$. We prove the bi-implication by structural induction on $[\varphi]$.

\underline{Forward Direction:} Given $(w_d,[t]) \models [\varphi]$, we need to show that $(\sigma,j) \models \tau([\varphi])$. 

\begin{itemize}
    \item \textbf{Case 1: $[\varphi]=\top$.} 
    
    By Definition~\ref{def:SSTL_boolean_semantics}, $(w_d,[t])\models \top$ holds trivially. Similarly, by Definition~\ref{def:ltlp_semantics}, $(\sigma,j)\models \top$ also holds trivially. Therefore,
    \[
    (w_d,[t])\models\top \implies (\sigma,j)\models\tau(\top).
    \]

    \item \textbf{Case 2: $[\varphi]=x_i^{w_d}([t])\geq0$.} 
    
    By Definition~\ref{def:ltlp_sstl_time_bijection}, the translation is $\tau\bigl(x_i^{w_d}([t])\ge0\bigr)=s_{j}(x_i^{w_d})\ge0$. Given that $(w_d,[t])\models x_i^{w_d}([t])\geq0$, we have $x_i^{w_d}([t])\geq0$ by the SSTL semantics (Definition~\ref{def:SSTL_boolean_semantics}). Since the state $s_{j}$ in the LTL$_{\mathcal P}$ trace assigns the same value (where $j=[t]$), $s_{j}(x_i^{w_d})=x_i^{w_d}([t])$, we obtain $s_{j}(x_i^{w_d})\geq0$. Thus, by the LTL$_{\mathcal P}$ semantics (Definition~\ref{def:ltlp_semantics}), $(\sigma,j)\models s_{j}(x_i^{w_d})\ge0$.

\item \textbf{Case 3: $[\varphi]=\neg[\psi]$.} 

Assume $(w_d,[t])\models\neg[\psi]$.
\begin{itemize}
    \item By SSTL semantics (Definition~\ref{def:SSTL_boolean_semantics}), this means $(w_d,[t])\not\models[\psi]$.
    \item By the structural induction hypothesis on the subformula $[\psi]$, we have that $(w_d,[t])\not\models[\psi]$ implies $(\sigma,j)\not\models\tau([\psi])$.
    \item By LTL$_{\mathcal P}$ semantics (Definition~\ref{def:ltlp_semantics}), $(\sigma,j)\not\models\tau([\psi])$ means $(\sigma,j)\models\neg\tau([\psi])$.
    \item By the translation (Definition~\ref{def:ltlp_sstl_time_bijection}), $\tau(\neg[\psi])=\neg\tau([\psi])$.
\end{itemize}
Therefore, $(w_d,[t])\models\neg[\psi] \implies (\sigma,j)\models\tau(\neg[\psi])$.

\item \textbf{Case 4: $[\varphi]=[\psi_1]\land[\psi_2]$.} 

Assume $(w_d,[t])\models[\psi_1]\land[\psi_2]$.
\begin{itemize}
    \item By SSTL semantics (Definition~\ref{def:SSTL_boolean_semantics}), this means $(w_d,[t])\models[\psi_1]$ and $(w_d,[t])\models[\psi_2]$.
    \item By the structural induction hypothesis, $(w_d,[t])\models[\psi_1]$ implies $(\sigma,j)\models\tau([\psi_1])$ and $(w_d,[t])\models[\psi_2]$ implies $(\sigma,j)\models\tau([\psi_2])$.
    \item By LTL$_{\mathcal P}$ semantics (Definition~\ref{def:ltlp_semantics}), $(\sigma,j)\models\tau([\psi_1])$ and $(\sigma,j)\models\tau([\psi_2])$ means $(\sigma,j)\models\tau([\psi_1])\land\tau([\psi_2])$.
    \item By the translation (Definition~\ref{def:ltlp_sstl_time_bijection}), $\tau([\psi_1]\land[\psi_2])=\tau([\psi_1])\land\tau([\psi_2])$.
\end{itemize}
Therefore, $(w_d,[t])\models[\psi_1]\land[\psi_2] \implies (\sigma,j)\models\tau([\psi_1]\land[\psi_2])$.

\item \textbf{Case 5: $[\varphi]=[\psi_1]\,\mathcal U\,[\psi_2]$ (unbounded Until).}

Assume $(w_d,[t])\models[\psi_1]\,\mathcal U\,[\psi_2]$.
\begin{itemize}
    \item By SSTL semantics (Definition~\ref{def:SSTL_boolean_semantics}), there exists a discrete time $[t']\geq[t]$ such that $(w_d,[t'])\models[\psi_2]$ and for all $[t'']\in[[t],[t'])$, $(w_d,[t''])\models[\psi_1]$.
    \item By the structural induction hypothesis, $(w_d,[t'])\models[\psi_2]$ implies $(\sigma,[t'])\models\tau([\psi_2])$ and for all $[t'']\in[[t],[t'])$, $(w_d,[t''])\models[\psi_1]$ implies $(\sigma,[t''])\models\tau([\psi_1])$.
    \item Since $j=[t]$ and by the correspondence established in Definition~\ref{def:ltlp_sstl_time_bijection}, there exists $k=[t']\geq j$ such that $(\sigma,k)\models\tau([\psi_2])$ and for all $l\in[j,k]$, $(\sigma,l)\models\tau([\psi_1])$.
    \item This is precisely the LTL$_{\mathcal P}$ semantics of Until (Definition~\ref{def:ltlp_semantics}): $(\sigma,j)\models\tau([\psi_1])\mathcal U\tau([\psi_2])$.
\end{itemize}
Therefore, $(w_d,[t])\models[\psi_1]\,\mathcal U\,[\psi_2] \implies (\sigma,j)\models\tau([\psi_1]\,\mathcal U\,[\psi_2])$.

\item \textbf{Case 6: $[\varphi]=[\psi_1]\,\mathcal U_{[[a],[b]]}\,[\psi_2]$ (bounded Until).}

Assume $(w_d,[t])\models[\psi_1]\,\mathcal U_{[[a],[b]]}\,[\psi_2]$.
\begin{itemize}
    \item By SSTL semantics (Definition~\ref{def:SSTL_boolean_semantics}), there exists a discrete time $[t']\in[[t+a],[t+b]]$ such that $(w_d,[t'])\models[\psi_2]$ and for all $[t'']\in[[t],[t'])$, $(w_d,[t''])\models[\psi_1]$.
    \item By the structural induction hypothesis, $(w_d,[t'])\models[\psi_2]$ implies $(\sigma,[t'])\models\tau([\psi_2])$ and for all $[t'']\in[[t],[t'])$, $(w_d,[t''])\models[\psi_1]$ implies $(\sigma,[t''])\models\tau([\psi_1])$.
    \item Let $j_0=[t]$ be the obligation start index. Since $[t']\in[[t+a],[t+b]]$, we have $[t+a]\leq[t']\leq[t+b]$, which means $j_0+a\leq j=[t']\leq j_0+b$.
    \item By the definition of the $\textsf{within}[a,b]$ guard, this means $\textsf{within}[a,b]$ is true at position $j=[t']$.
    \item Therefore, $(\sigma,[t'])\models\tau([\psi_2])\land\textsf{within}[a,b]$.
    \item Combining with the fact that for all $l\in[j_0,[t'])$, $(\sigma,l)\models\tau([\psi_1])$, we have by the LTL$_{\mathcal P}$ semantics of Until: $(\sigma,j_0)\models\tau([\psi_1])\mathcal U(\tau([\psi_2])\land\textsf{within}[a,b])$.
    \item By the translation, $\tau([\psi_1]\mathcal U_{[[a],[b]]}[\psi_2])=\tau([\psi_1])\mathcal U(\tau([\psi_2])\land\textsf{within}[a,b])$.
\end{itemize}
Therefore, $(w_d,[t])\models[\psi_1]\,\mathcal U_{[[a],[b]]}\,[\psi_2] \implies (\sigma,j)\models\tau([\psi_1]\,\mathcal U_{[[a],[b]]}\,[\psi_2])$.

\item \textbf{Case 7: $[\varphi]=\lozenge_{[[a],[b]]}[\psi]$ (bounded eventually).}

Assume $(w_d,[t])\models\lozenge_{[[a],[b]]}[\psi]$.
\begin{itemize}
    \item By the definition of bounded eventually in SSTL (Definition~\ref{def:SSTL_boolean_semantics}), this is equivalent to $\top\mathcal U_{[[a],[b]]}[\psi]$.
    \item By SSTL semantics, there exists a discrete time $[t']\in[[t+a],[t+b]]$ such that $(w_d,[t'])\models[\psi]$.
    \item By the structural induction hypothesis, $(w_d,[t'])\models[\psi]$ implies $(\sigma,[t'])\models\tau([\psi])$.
    \item Let $j_0=[t]$. Since $[t']\in[[t+a],[t+b]]$, we have $j_0+a\leq[t']\leq j_0+b$, which means $\textsf{within}[a,b]$ is true at position $[t']$.
    \item Therefore, $(\sigma,[t'])\models\tau([\psi])\land\textsf{within}[a,b]$.
    \item Since there exists such a position $[t']\geq[t]$, by the LTL$_{\mathcal P}$ semantics of Diamond: $(\sigma,j)\models\Diamond(\tau([\psi])\land\textsf{within}[a,b])$.
    \item By the translation (Definition~\ref{def:ltlp_sstl_time_bijection}), $\tau(\lozenge_{[[a],[b]]}[\psi])=\Diamond(\tau([\psi])\land\textsf{within}[a,b])$.
\end{itemize}
Therefore, $(w_d,[t])\models\lozenge_{[[a],[b]]}[\psi] \implies (\sigma,j)\models\tau(\lozenge_{[[a],[b]]}[\psi])$.

\item \textbf{Case 8: $[\varphi]=\square_{[[a],[b]]}[\psi]$ (bounded always).}

Assume $(w_d,[t])\models\square_{[[a],[b]]}[\psi]$.
\begin{itemize}
    \item By the definition of bounded always in SSTL (Definition~\ref{def:SSTL_boolean_semantics}), this is equivalent to $\neg\lozenge_{[[a],[b]]}\neg[\psi]$.
    \item By SSTL semantics, for all discrete times $[t']\in[[t+a],[t+b]]$, we have $(w_d,[t'])\models[\psi]$.
    \item By the structural induction hypothesis, for all $[t']\in[[t+a],[t+b]]$, $(w_d,[t'])\models[\psi]$ implies $(\sigma,[t'])\models\tau([\psi])$.
    \item Let $j_0=[t]$. For any position $k\geq j_0$, if $\textsf{within}[a,b]$ is true at $k$, then by definition $j_0+a\leq k\leq j_0+b$, which means $k=[t']$ for some $[t']\in[[t+a],[t+b]]$.
    \item For such $k$, we have $(\sigma,k)\models\tau([\psi])$.
    \item Therefore, for all $k\geq j_0$, if $\textsf{within}[a,b]$ holds at $k$, then $\tau([\psi])$ holds at $k$. This is the implication $\textsf{within}[a,b]\rightarrow\tau([\psi])$ at each position.
    \item By the LTL$_{\mathcal P}$ semantics of Box: $(\sigma,j_0)\models\Box(\textsf{within}[a,b]\rightarrow\tau([\psi]))$.
    \item By the translation (Definition~\ref{def:ltlp_sstl_time_bijection}), $\tau(\square_{[[a],[b]]}[\psi])=\Box(\textsf{within}[a,b]\rightarrow\tau([\psi]))$.
\end{itemize}
Therefore, $(w_d,[t])\models\square_{[[a],[b]]}[\psi] \implies (\sigma,j)\models\tau(\square_{[[a],[b]]}[\psi])$.

\end{itemize}

Thus, the forward direction holds for all SSTL formulae.

\underline{Backward Direction:} Given $(\sigma,j) \models \tau([\varphi])$, we need to show that $(w_d,[t]) \models [\varphi]$. 

\begin{itemize}
    \item \textbf{Case 1: $[\varphi]=\top$.} 
    
    By Definition~\ref{def:ltlp_semantics}, $(\sigma,j)\models \top$ holds trivially. Similarly, by Definition~\ref{def:SSTL_boolean_semantics}, $(w_d,[t])\models \top$ also holds trivially. Therefore,
    \[
    (\sigma,j)\models\tau(\top) \implies (w_d,[t])\models\top.
    \]

    \item \textbf{Case 2: $[\varphi]=x_i^{w_d}([t])\geq0$.} 
    
    Assume $(\sigma,j)\models s_j(x_i^{w_d})\ge0$.
    \begin{itemize}
        \item By LTL$_{\mathcal P}$ semantics (Definition~\ref{def:ltlp_semantics}), this means $s_j(x_i^{w_d})\geq0$.
        \item By the correspondence in Definition~\ref{def:ltlp_sstl_time_bijection}, $s_j(x_i^{w_d})=x_i^{w_d}([t])$.
        \item Therefore, $x_i^{w_d}([t])\geq0$.
        \item By SSTL semantics (Definition~\ref{def:SSTL_boolean_semantics}), $(w_d,[t])\models x_i^{w_d}([t])\geq0$.
    \end{itemize}
    Therefore, $(\sigma,j)\models\tau(x_i^{w_d}([t])\geq0) \implies (w_d,[t])\models x_i^{w_d}([t])\geq0$.

\item \textbf{Case 3: $[\varphi]=\neg[\psi]$.} 

Assume $(\sigma,j)\models\tau(\neg[\psi])$.
\begin{itemize}
    \item By the translation, $\tau(\neg[\psi])=\neg\tau([\psi])$.
    \item By LTL$_{\mathcal P}$ semantics (Definition~\ref{def:ltlp_semantics}), $(\sigma,j)\models\neg\tau([\psi])$ means $(\sigma,j)\not\models\tau([\psi])$.
    \item By the structural induction hypothesis on the subformula $[\psi]$, $(\sigma,j)\not\models\tau([\psi])$ implies $(w_d,[t])\not\models[\psi]$.
    \item By SSTL semantics (Definition~\ref{def:SSTL_boolean_semantics}), $(w_d,[t])\not\models[\psi]$ means $(w_d,[t])\models\neg[\psi]$.
\end{itemize}
Therefore, $(\sigma,j)\models\tau(\neg[\psi]) \implies (w_d,[t])\models\neg[\psi]$.

\item \textbf{Case 4: $[\varphi]=[\psi_1]\land[\psi_2]$.} 

Assume $(\sigma,j)\models\tau([\psi_1]\land[\psi_2])$.
\begin{itemize}
    \item By the translation, $\tau([\psi_1]\land[\psi_2])=\tau([\psi_1])\land\tau([\psi_2])$.
    \item By LTL$_{\mathcal P}$ semantics (Definition~\ref{def:ltlp_semantics}), $(\sigma,j)\models\tau([\psi_1])\land\tau([\psi_2])$ means $(\sigma,j)\models\tau([\psi_1])$ and $(\sigma,j)\models\tau([\psi_2])$.
    \item By the structural induction hypothesis, $(\sigma,j)\models\tau([\psi_1])$ implies $(w_d,[t])\models[\psi_1]$ and $(\sigma,j)\models\tau([\psi_2])$ implies $(w_d,[t])\models[\psi_2]$.
    \item By SSTL semantics (Definition~\ref{def:SSTL_boolean_semantics}), $(w_d,[t])\models[\psi_1]$ and $(w_d,[t])\models[\psi_2]$ means $(w_d,[t])\models[\psi_1]\land[\psi_2]$.
\end{itemize}
Therefore, $(\sigma,j)\models\tau([\psi_1]\land[\psi_2]) \implies (w_d,[t])\models[\psi_1]\land[\psi_2]$.

\item \textbf{Case 5: $[\varphi]=[\psi_1]\,\mathcal U\,[\psi_2]$ (unbounded Until).}

Assume $(\sigma,j)\models\tau([\psi_1]\mathcal U[\psi_2])$.
\begin{itemize}
    \item By the translation, $\tau([\psi_1]\mathcal U[\psi_2])=\tau([\psi_1])\mathcal U\tau([\psi_2])$.
    \item By LTL$_{\mathcal P}$ semantics (Definition~\ref{def:ltlp_semantics}), there exists $k\geq[t]$ such that $(\sigma,k)\models\tau([\psi_2])$ and for all $l\in[[t],k)$, $(\sigma,l)\models\tau([\psi_1])$.
    \item Let $[t']:=k$. By the structural induction hypothesis, $(\sigma,k)\models\tau([\psi_2])$ implies $(w_d,[t'])\models[\psi_2]$ and for all $l\in[[t],k)$, $(\sigma,l)\models\tau([\psi_1])$ implies $(w_d,l)\models[\psi_1]$.
    \item Since $[t']=k\geq[t]$, we have $[t']\geq[t]$.
    \item Therefore, there exists $[t']\geq[t]$ such that $(w_d,[t'])\models[\psi_2]$ and for all $[t'']\in[[t],[t'])$, $(w_d,[t''])\models[\psi_1]$.
    \item By SSTL semantics (Definition~\ref{def:SSTL_boolean_semantics}), this is precisely $(w_d,[t])\models[\psi_1]\mathcal U[\psi_2]$.
\end{itemize}
Therefore, $(\sigma,j)\models\tau([\psi_1]\mathcal U[\psi_2]) \implies (w_d,[t])\models[\psi_1]\mathcal U[\psi_2]$.

\item \textbf{Case 6: $[\varphi]=[\psi_1]\,\mathcal U_{[[a],[b]]}\,[\psi_2]$ (bounded Until).}

Assume $(\sigma,j)\models\tau([\psi_1]\mathcal U_{[[a],[b]]}[\psi_2])$.
\begin{itemize}
    \item By the translation, $\tau([\psi_1]\mathcal U_{[[a],[b]]}[\psi_2])=\tau([\psi_1])\mathcal U(\tau([\psi_2])\land\textsf{within}[a,b])$.
    \item By LTL$_{\mathcal P}$ semantics (Definition~\ref{def:ltlp_semantics}), there exists $k\geq[t]$ such that $(\sigma,k)\models\tau([\psi_2])\land\textsf{within}[a,b]$ and for all $l\in[[t],k)$, $(\sigma,l)\models\tau([\psi_1])$.
    \item From $(\sigma,k)\models\tau([\psi_2])\land\textsf{within}[a,b]$, we have $(\sigma,k)\models\tau([\psi_2])$ and $(\sigma,k)\models\textsf{within}[a,b]$.
    \item Let $j_0=[t]$ and $[t']:=k$. Since $(\sigma,k)\models\textsf{within}[a,b]$, by the definition of the guard we have $j_0+a\leq k\leq j_0+b$, which means $[t+a]\leq[t']\leq[t+b]$, so $[t']\in[[t+a],[t+b]]$.
    \item By the structural induction hypothesis, $(\sigma,k)\models\tau([\psi_2])$ implies $(w_d,[t'])\models[\psi_2]$ and for all $l\in[[t],k)$, $(\sigma,l)\models\tau([\psi_1])$ implies $(w_d,l)\models[\psi_1]$.
    \item Therefore, there exists $[t']\in[[t+a],[t+b]]$ such that $(w_d,[t'])\models[\psi_2]$ and for all $[t'']\in[[t],[t'])$, $(w_d,[t''])\models[\psi_1]$.
    \item By SSTL semantics (Definition~\ref{def:SSTL_boolean_semantics}), this is precisely $(w_d,[t])\models[\psi_1]\mathcal U_{[[a],[b]]}[\psi_2]$.
\end{itemize}
Therefore, $(\sigma,j)\models\tau([\psi_1]\mathcal U_{[[a],[b]]}[\psi_2]) \implies (w_d,[t])\models[\psi_1]\mathcal U_{[[a],[b]]}[\psi_2]$.

\item \textbf{Case 7: $[\varphi]=\lozenge_{[[a],[b]]}[\psi]$ (bounded eventually).}

Assume $(\sigma,j)\models\tau(\lozenge_{[[a],[b]]}[\psi])$.
\begin{itemize}
    \item By the translation, $\tau(\lozenge_{[[a],[b]]}[\psi])=\Diamond(\tau([\psi])\land\textsf{within}[a,b])$.
    \item By LTL$_{\mathcal P}$ semantics (Definition~\ref{def:ltlp_semantics}), there exists $k\geq[t]$ such that $(\sigma,k)\models\tau([\psi])\land\textsf{within}[a,b]$.
    \item From this, we have $(\sigma,k)\models\tau([\psi])$ and $(\sigma,k)\models\textsf{within}[a,b]$.
    \item Let $j_0=[t]$ and $[t']:=k$. Since $(\sigma,k)\models\textsf{within}[a,b]$, by the definition of the guard we have $j_0+a\leq k\leq j_0+b$, which means $[t+a]\leq[t']\leq[t+b]$, so $[t']\in[[t+a],[t+b]]$.
    \item By the structural induction hypothesis, $(\sigma,k)\models\tau([\psi])$ implies $(w_d,[t'])\models[\psi]$.
    \item Therefore, there exists $[t']\in[[t+a],[t+b]]$ such that $(w_d,[t'])\models[\psi]$.
    \item By the definition of bounded eventually in SSTL (Definition~\ref{def:SSTL_boolean_semantics}), this is precisely $(w_d,[t])\models\lozenge_{[[a],[b]]}[\psi]$.
\end{itemize}
Therefore, $(\sigma,j)\models\tau(\lozenge_{[[a],[b]]}[\psi]) \implies (w_d,[t])\models\lozenge_{[[a],[b]]}[\psi]$.

\item \textbf{Case 8: $[\varphi]=\square_{[[a],[b]]}[\psi]$ (bounded always).}

Assume $(\sigma,j)\models\tau(\square_{[[a],[b]]}[\psi])$.
\begin{itemize}
    \item By the translation, $\tau(\square_{[[a],[b]]}[\psi])=\Box(\textsf{within}[a,b]\rightarrow\tau([\psi]))$.
    \item By LTL$_{\mathcal P}$ semantics (Definition~\ref{def:ltlp_semantics}), for all $k\geq[t]$, we have $(\sigma,k)\models\textsf{within}[a,b]\rightarrow\tau([\psi])$.
    \item This means for all $k\geq[t]$, if $(\sigma,k)\models\textsf{within}[a,b]$, then $(\sigma,k)\models\tau([\psi])$.
    \item Consider an arbitrary $[t']\in[[t+a],[t+b]]$. Let $k:=[t']$. Then $[t+a]\leq k\leq[t+b]$, which means $k$ satisfies $j_0+a\leq k\leq j_0+b$ where $j_0=[t]$.
    \item Therefore, $(\sigma,k)\models\textsf{within}[a,b]$.
    \item By the implication above, $(\sigma,k)\models\tau([\psi])$.
    \item By the structural induction hypothesis, $(\sigma,k)\models\tau([\psi])$ implies $(w_d,[t'])\models[\psi]$.
    \item Since $[t']$ was arbitrary in $[[t+a],[t+b]]$, for all $[t']\in[[t+a],[t+b]]$, we have $(w_d,[t'])\models[\psi]$.
    \item By the definition of bounded always in SSTL (Definition~\ref{def:SSTL_boolean_semantics}), this is precisely $(w_d,[t])\models\square_{[[a],[b]]}[\psi]$.
\end{itemize}
Therefore, $(\sigma,j)\models\tau(\square_{[[a],[b]]}[\psi]) \implies (w_d,[t])\models\square_{[[a],[b]]}[\psi]$.

\end{itemize}

Thus, the backward direction holds for all SSTL formulae.

%-------------------------------------------------
Combining both directions, we have shown that for every SSTL formula $[\varphi]$ and every discrete time $[t]$,
\[(w_d,[t]) \models [\varphi]\;\iff\;(\sigma,j) \models \tau([\varphi]),\]
which completes the proof of Theorem~\ref{theorem:sstl_to_ltlp_correctness}.

\end{proof}

\section{Proof of Lemma \ref{lemma:boundedness-j0}}
\label{appendix:boundedness_j0}
\begin{proof}[Proof sketch]
  The core of this proof is that each time-bounded temporal operator $\beta$ in the SSTL
  formula contributes obligations whose lifetimes are limited to the width of
  its interval.  Let $w_\beta=b_\beta-a_\beta+1$ be that width and let
  $W=\sum_{\beta} w_\beta$ (finite because the formula is finite).

  At any discrete position $j$ we track, for every active obligation, the entry
  index $j_0$ captured when the operator became enabled.  Two observations
  are important:
  \begin{enumerate}
    \item \textbf{Lifetime bound.}  An obligation created at $j_0$ for operator
          $\beta$ can only stay alive while $j$ satisfies $j_0+a_\beta\le j\le
          j_0+b_\beta$; hence its lifetime is at most $w_\beta$ steps.
    \item \textbf{Creation rate.}  Per operator, at most one new obligation can
          start at each position.
  \end{enumerate}
  Combining the two facts, within any sliding window of length $w_\beta$ there
  can be at most $w_\beta$ simultaneously alive obligations for $\beta$.
  Summing over all operators yields the uniform bound $W$ on the number of
  distinct $j_0$ values that must be tracked at any position.  This finiteness
  underpins the decidability result in
  Lemma~\ref{remark:decidability}.\qedhere
\end{proof}

\section{Implementation Details}
\label{appendix:implementation_details}

Handling of bounded time intervals in the translation is not unique for SSTL to LTL$_{\mathcal{P}}$ translation and other equivalent translations may be used. Here, we present one such translation that is implementable and faster to verify.

The $\textsf{within}[a,b]$ condition used in the translation is conceptually and logically correct. However, it is not easily implementable for verification purposes due to how SPIN handles state space exploration in LTL formulae. To address this, we can modify the translation to an equiavalent translation that is implementable and (possibly) faster to verify as follows:

\begin{align}
  \tau\bigl([\varphi_1]\;\mathcal U_{[[a],[b]]}\;[\varphi_2]\bigr) &:= \tau([\varphi_1]) \land (j \leq j_0 + b) \;\mathcal U\; \tau([\varphi_2]) \land (j \geq j_0 + a) \\
  \tau\bigl(\lozenge_{[[a],[b]]}[\varphi]\bigr) &:= \tau\bigl(\top\;\mathcal U_{[[a],[b]]}\;[\varphi]\bigr) \\
  \tau\bigl(\square_{[[a],[b]]}[\varphi]\bigr) &:= \neg \lozenge_{[[a],[b]]}\neg[\varphi]
\end{align}

We use this translation in the implementation of the SSTL to LTL$_{\mathcal{P}}$ translation.

\begin{proof}[Proof sketch of equivalence]
  We compare the original guard $\textsf{within}[a,b] := (j_0+a \le j \le j_0+b)$
  with the implementable encodings used in the three translated cases.
  For each operator we show that the set of traces accepted by the SPIN–friendly
  formula is identical to that accepted by the specification using
  $\textsf{within}[a,b]$.
  \begin{enumerate}
    \item \textbf{Bounded Until.}  At any position $j$ the obligation should
          succeed when some future $j'$ satisfies $j_0+a \le j' \le j_0+b$ and
          $\varphi_2$ holds, while $\varphi_1$ holds at all ticks prior to $j'$.
          Re-writing this with the ordinary Until yields exactly the LTL formula
          $\tau([\varphi_1]) \land (j \le j_0+b)$ $\;\mathcal U\;$
          $\tau([\varphi_2]) \land (j \ge j_0+a)$.  The left conjunct ensures
          the path has not yet stepped beyond $j_0+b$ (matching the upper bound
          of the window) while the right conjunct enforces that the witness tick
          occurs no earlier than $j_0+a$.  Hence the temporal window is
          enforced equivalently.
    \item \textbf{Bounded Eventually and Always.} This is a direct consequence of the bounded Until case.
    \end{enumerate}

  In every case the implementable encoding enforces the same lower and upper
  bounds on $j$ as the original $\textsf{within}[a,b]$ predicate, establishing
  semantic equivalence.\qedhere
\end{proof}

\end{document}